\documentclass[11pt,twoside]{article}

\usepackage[hmargin=1.25in,vmargin=1in]{geometry}
\usepackage[utf8]{inputenc}
\usepackage{graphicx,amsmath,stmaryrd,cite,enumerate}
\usepackage{jeffe}

\numberwithin{figure}{section}
\numberwithin{equation}{section}
\newtheorem{theorem}{Theorem}[section]
\newtheorem{lemma}[theorem]{Lemma}
\newtheorem{claim}[theorem]{Claim}
\newtheorem{corollary}[theorem]{Corollary}
\newtheorem{conjecture}[theorem]{Conjecture}

\usepackage[charter]{mathdesign}
\usepackage{sourcesanspro,inconsolata,eucal,microtype}

\usepackage{caption}
\newif\ifneedsubqed
\global\needsubqedfalse
\def\subQED{\ensuremath{{\lozenge}}}
\def\subqed{\ifneedsubqed\markatright{\subQED}\global\needsubqedfalse\else\null\fi}
\newproof{subproof}{Proof\global\needsubqedtrue}[\subqed]

\def\EMPH#1{\textbf{\textcolor{BrickRed}{\emph{\boldmath #1}}}}
\def\coloneq{\mathrel{\mathop:}=}	% := with centered :

\def\Map{\Gamma}
\def\FaceUnion{U}
\def\Sleeve{S}
\def\Sbdry{\sigma}

\def\South#1{{#1}'}
\def\NewSouth#1{{#1}''}

\def\InitMap{\Map}
\def\SouthMap{\Map'}
\def\NewSouthMap{\Map''}

\def\InitZ{z}
\def\SouthZ{z'}
\def\NewSouthZ{z''}

\def\Bar{\beta}
\def\PreBar{B}

\def\NorthF{f_{\textsf{north}}}
\def\AboveF#1{f_{\textsf{above}}(#1)}
\def\BelowF#1{f_{\textsf{below}}(#1)}
\def\SouthF{f_{\textsf{south}}}

\def\Dn#1{{#1}^{\shortdownarrow}}
\def\Rt#1{{#1}^{\shortrightarrow}}
\def\DnVee#1{{#1}^{\scriptscriptstyle\boldsymbol{\ssearrow\!\sswarrow}}}

\def\BDn#1{{#1}^{\pmb{\shortdownarrow}}}
\def\BRt#1{{#1}^{\pmb{\shortrightarrow}}}
\def\BDnVee#1{{#1}^{\scriptscriptstyle\pmb{\ssearrow\!\sswarrow}}}

\title{Shelling and Sinking Graphs on the Sphere%
\thanks{An extended abstract of this paper was presented at the 41st International Symposium on Computational Geometry~\cite{eh-ssgs-25}.}}

\author{%
\href{http://jeffe.cs.illinois.edu}{Jeff Erickson}%\orcidID{0000-0002-5253-2282}
\qquad\qquad
\href{https://christianjhoward.me/}{Christian Howard}%\orcidID{0009-0009-5813-3643}%
\\[1ex]
University of Illinois Urbana-Champaign}

\date{August 1, 2025}  % CORRECT DATE BEFORE RE-UPLOADING TO ARXIV!

\begin{document}
%\begin{titlepage}
\maketitle

\begin{abstract}\noindent
We describe a promising approach to efficiently morph shortest-path embeddings of planar graphs on the sphere, extending earlier approaches of Awartani and Henderson [\emph{Trans.~AMS} 1987] and Kobourov and Landis [\emph{JGAA} 2006]. Specifically, we describe two methods to morph spherical triangulations by moving their vertices along longitudes into the southern hemisphere; we call a triangulation \emph{sinkable} if such a morph exists.  Our first method generalizes a longitudinal shelling construction of Awartani and Henderson; a triangulation is sinkable if a specific orientation of its dual graph is acyclic.  We describe a simple polynomial-time algorithm to find a longitudinally shellable rotation of a given spherical triangulation, if one exists; we also construct a spherical triangulation that has no longitudinally shellable rotation.  Our second method is based on a linear-programming characterization of sinkability.  By identifying its optimal basis, we show that this linear program can be solved in $O(n^{\omega/2})$ time, where $\omega$ is the matrix-multiplication exponent, assuming the underlying linear system is non-singular.  In addition to these main results, we describe a reduction from morphing shortest-path embeddings of 3-connected planar graphs on the sphere to morphing triangulations, and we describe an efficient algorithm that constructs morphs where each intermediate edge has at most one bend.  Finally, we pose several conjectures and describe experimental results that support them.
\end{abstract}

%\end{titlepage}
	
% ------------------------------------------------------------------------------
\section{Introduction}

A \emph{morph} between two planar straight-line graphs is a continuous deformation from one to the other, such that all edges in all intermediate graphs are interior-disjoint line segments.  Planar graph morphing has many applications in graphics, animation, visualization, and geometric modeling, as well as connections to fundamental questions in low-dimensional topology.

There are two state-of-the-art approaches for morphing planar graphs.  The first is the barycentric interpolation method of Floater, Gotsman, and Surazhsky \cite{fg-mti-99,gs-gipm-01,sg-cmcpt-01,sg-msfuo-01,sg-imct-03,ekp-ifmpg-03}, which is based on an extension by Floater \cite{f-psast-97,f-ptsda-98} of Tutte's classical spring-embedding theorem \cite{t-hdg-63}.  These algorithms compute an implicit representation of a smooth morph, any intermediate stage of which can be constructed in $O(n^{\omega/2}) = O(n^{1.1864})$ time, where $\omega < 2.3728$ is the matrix-multiplication exponent.  The second method combines an edge-collapsing strategy originally proposed by Cairns \cite{c-dprc-44,c-idgc2-44,h-dphcp-74,t-dpg-83} with more recent algorithms for convex hierarchical planar graph drawing \cite{cgt-cdg23-96,hn-cdhpg-10,c-abcht-19,kklss-cimpg-19,k-cdhgl-21}.  This method has led to several efficient algorithms for constructing explicit representations of piecewise-linear morphs \cite{aabcd-hmpgd-17,kklss-cimpg-19,k-cdhgl-21,el-ptmme-23}. %
The fastest of these algorithms, due to Klemz \cite{k-cdhgl-21}, computes a morph consisting of $O(n)$ stages in $O(n^2)$ time, where in each stage, all vertices move along parallel lines at constant speeds.  A recent algorithm of Erickson and Lin computes explicit piecewise-linear morphs using barycentric interpolation in $O(n^{1+\omega/2})$ time \cite{el-ptmme-23}.
Algorithms are also known for several variants of planar morphing \cite{lp-mpgdb-11, bhl-msdpt-19, bls-mpgdo-24, ckkrw-mrd-23, gv-ompod-18, befklow-mpgdt-23}.

Both planar morphing approaches have recently been generalized to geodesic graphs on the flat torus \cite{celp-hmgt-21,el-ptmme-23,lwz-dsgtf-21}; even more recently, Luo, Wu, and Zhu generalized the barycentric interpolation method to arbitrary negative-curvature surfaces of higher genus \cite{lwz-dsgtg-23,lwz-sgtc-24}. 

\subsection{What About the Sphere?}

In light of the long history of results for morphing planar graphs and their recent generalizations to higher-genus surfaces, it is surprising how little is known about morphing graphs on the sphere.  Although embeddings on the sphere are \emph{topologically} equivalent to embeddings in the plane, the different geometries of the two surfaces induce significantly different behavior.

In 1944, Cairns \cite{c-idgc2-44} proved that any two isomorphic shortest-path triangulations of the sphere are connected by a continuous family of shortest-path triangulations, using essentially the same edge-contraction strategy that he used for planar triangulations (but with more complex case analysis).  A direct translation of Cairns's spherical morphing proof leads to an exponential-time algorithm; unlike his planar result, this is still the fastest algorithm known.  Neither of the primary tools that underlie more efficient planar morphing algorithms---barycentric embeddings \cite{f-psast-97,f-ptsda-98,t-hdg-63} and convex hierarchical graph drawing \cite{cgt-cdg23-96,hn-cdhpg-10,c-abcht-19,kklss-cimpg-19,k-cdhgl-21}---have appropriate generalizations to spherical graphs. In fact, as far as we are aware, this is the \emph{only} algorithm known for morphing arbitrary spherical triangulations.  

More efficient morphing algorithms are known for a few special cases of spherical triangulations.  Almost all of these algorithms operate by moving vertices along longitudes into the southern hemisphere, projecting the resulting “southern” triangulation into the plane, and applying a planar morphing algorithm.  We call a sphere triangulation \EMPH{sinkable} if it can be morphed along longitudes into the southern hemisphere.

Awartani and Henderson \cite{ah-sgts-87} describe an algorithm to morph spherical triangulations that satisfy a certain longitudinal shelling condition using this strategy.  A spherical triangulation is \EMPH{longitudinally shellable} if its faces can be ordered so that the union of any prefix of faces has connected intersection with every longitude.   In short, Awartani and Henderson prove that every longitudinally shellable triangulation is sinkable.  They also prove that any triangulation with a \emph{longitudinal seam}, meaning there is a longitude that does not cross the interior of any edge, is longitudinally shellable.  (We study Awartani and Henderson's longitudinal shelling condition in more detail in Section~\ref{S:shelling}.)

Kobourov and Landis \cite{kl-mpgss-06} describe an algorithm to morph between \emph{Delaunay} triangulations via longitudinal morphs; their algorithm easily generalizes to \EMPH{coherent} triangulations, which are central projections of arbitrary convex polyhedra.  Every coherent triangulation is longitudinally shellable and therefore sinkable.  The same method is also implicit in Richter-Gebert's proof \cite[Theorem 13.3.3]{r-rsp-96} of classical theorems of Eberhard \cite{e-mp-1891} and Steinitz \cite{sr-vtp-34} describing morphs between isomorphic convex polyhedra.

Awartani and Henderson \cite[Figure 3.1]{ah-sgts-87} also give an example of a triangulation that is not sinkable; it is not possible to move the vertices of their triangulation into the southern hemisphere along longitudes without inverting at least one face.

\subsection{Simple Examples}
\label{SS:examples}

%\begin{figure}[htb]
%\centering
%\includegraphics[scale=0.3]{Fig/AH-turbine}
%\caption{Awartani and Henderson's unsinkable triangulation .}
%\end{figure}
  
The simplest longitudinally unshellable or unsinkable triangulations are central projections of sufficiently twisted \emph{Schönhardt polyhedra} \cite{s-uzdt-28}.  Let $S_0$ be an equilateral triangular prism, centered at the origin, whose triangular faces are parallel to the $xy$-plane, and whose quadrilateral faces are triangulated so that every vertex has degree $4$.  For any angle $\theta$, the Schönhardt polyhedron~$S_\theta$ is obtained by rotating the top triangular face of $S_0$ by $\theta$ and retaining the same facial structure.  Let~$\bar{S}_\theta$ denote the central projection of $S_\theta$ onto the unit sphere, as shown in Figure~\ref{F:schonhardt}.

When the twisting angle $\theta$ is negative, the Schönhardt polyhedron $S_\theta$ is convex, so its projection~$\bar{S}_\theta$ is coherent, and therefore longitudinally shellable, and therefore sinkable.  When $\theta$ is positive, the diagonals of the rectangular faces buckle inward, making the octahedron $S_\theta$ non-convex; Supnick proved that the resulting spherical triangulations $\bar{S}_\theta$ are the simplest non-coherent triangulations \cite{s-pdp-48, s-pdp2-51}. (See also Connelly and Henderson \cite{ch-c3sis-80} and De Loera, Rambau, and Santos \cite[Chapter 7.1]{lrs-tsaa-10}.)  Our Lemma~\ref{L:shellable} implies that $\bar{S}_\theta$ is not longitudinally shellable if $\theta>0$.  Finally, our algorithms in Section~\ref{S:sink} imply that $\bar{S}_\theta$ is sinkable if $0 < \theta < \pi/6$ but  unsinkable if $\theta \ge \pi/6$.

\begin{figure}[htb]
\centering\footnotesize\sffamily
\begin{tabular}{c@{\qquad}c@{\qquad}c}
\includegraphics[scale=0.4]{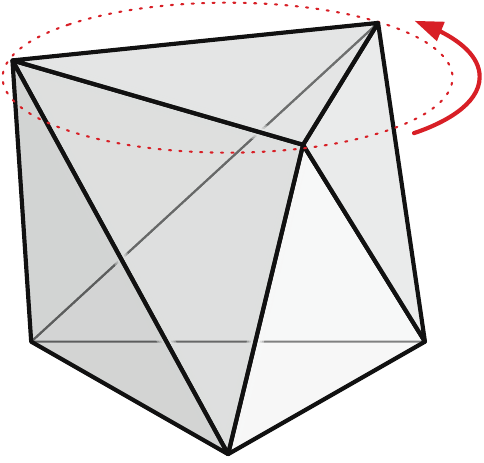} & 
\includegraphics[scale=0.4]{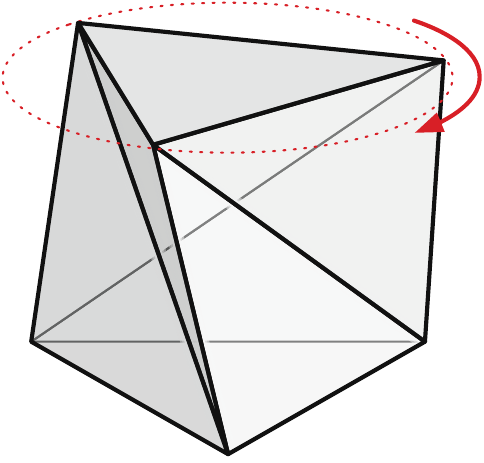} & 
\includegraphics[scale=0.4]{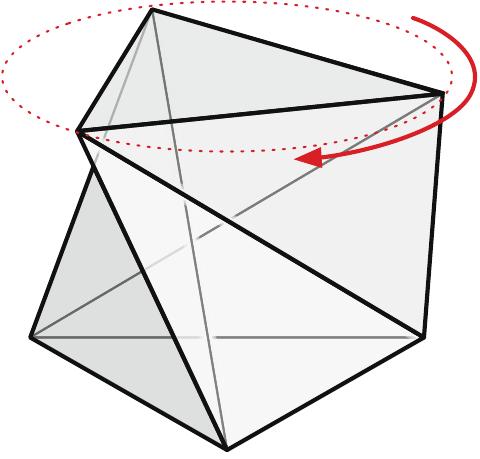} \\[1ex]
\includegraphics[scale=0.275,page=1]{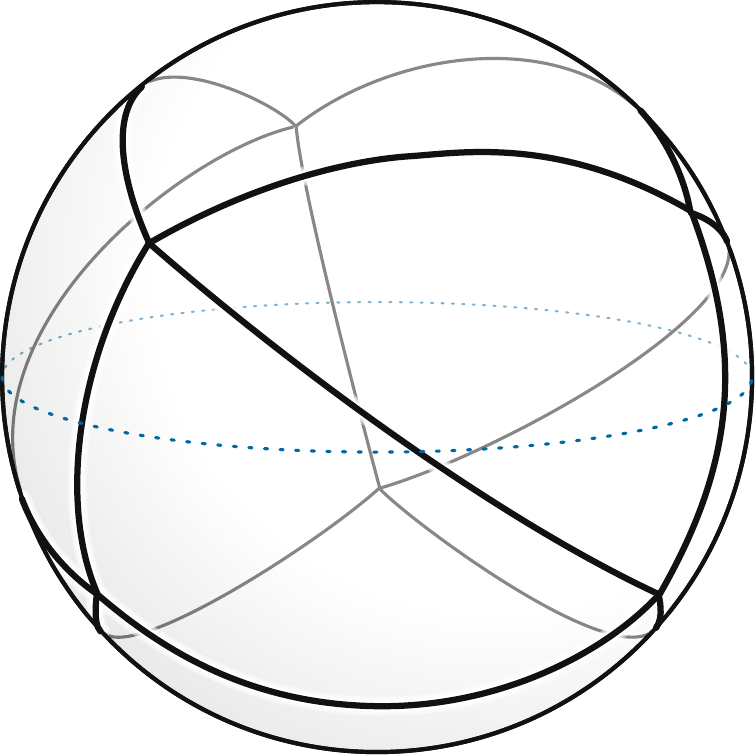} & 
\includegraphics[scale=0.275,page=2]{Fig/schonhardt-spheres} & 
\includegraphics[scale=0.275,page=3]{Fig/schonhardt-spheres} \\[1ex]
\includegraphics[scale=0.35]{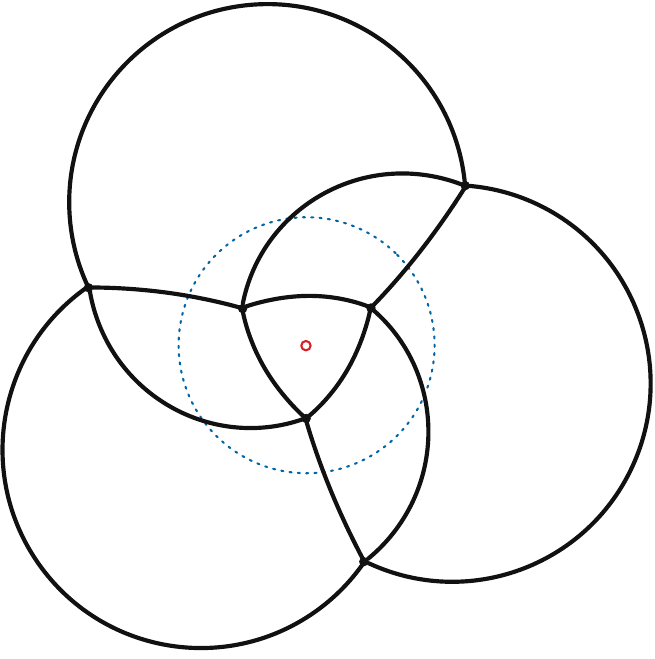} &
\includegraphics[scale=0.35]{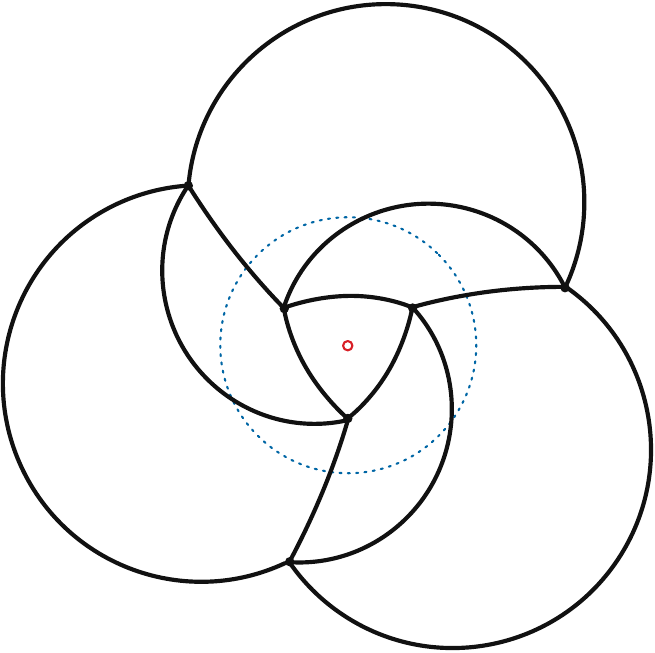} &
\includegraphics[scale=0.35]{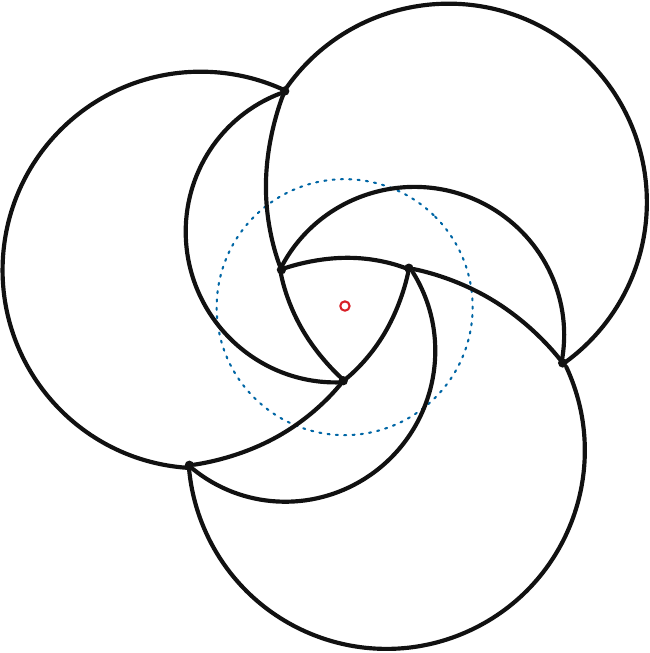} \\
 (a) & (b) & (c)
\end{tabular}
\caption{Three Schönhardt polyhedra, the corresponding spherical triangulations, and their stereographic projections into the plane. Depending on the twisting angle, these triangulations are (a) coherent and therefore longitudinally shellable, (b)~sinkable but not longitudinally shellable, or (c) not sinkable.}
\label{F:schonhardt}
\end{figure}

Another natural one-parameter family of twisted triangulations are projections of non-convex icosahedra that Douady dubbed \emph{six-beaked shaddocks} \cite{d-sb-71}; this family includes Jessen's orthogonal icosahedron \cite{j-oi-67} and is closely related to  Fuller's \emph{jitterbug} mechanism \cite{f-segt-75,e-j-87}.  Following Fuller, these icosahedra can be constructed by triangulating the square facets of a cuboctahedron so that every vertex has degree $5$, and then simultaneously twisting every equilateral triangular facet; the triangulated square facets buckle inward when the twisting angle is positive.  Every positively twisted shaddock triangulation is non-coherent \cite{d-sb-71}.  However, if one pair of equilateral triangular facets is normal to the $z$-axis, non-coherent shaddock triangulations can be longitudinally shellable, sinkable but longitudinally unshellable, or unsinkable, depending on the twisting angle.  See Figure \ref{F:morph-example}.%
\footnote{Schönhardt polyhedra and six-beaked shaddocks are also classical examples of polyhedra that cannot be triangulated without additional vertices \cite{s-uzdt-28,bc-np-16}.  The Schönhardt polyhedron $S_{\pi/6}$ and Jessen's orthogonal icosahedron are \emph{also} classical examples of infinitesimally flexible or “shaky” polyhedra \cite{w-skwba-65,gk-mfjoi-16,m-lbrcp-94,g-ups-78}.}  

\begin{figure}[htb]
\centering\footnotesize\sffamily
\begin{tabular}{cccc}
\includegraphics[scale=0.3]{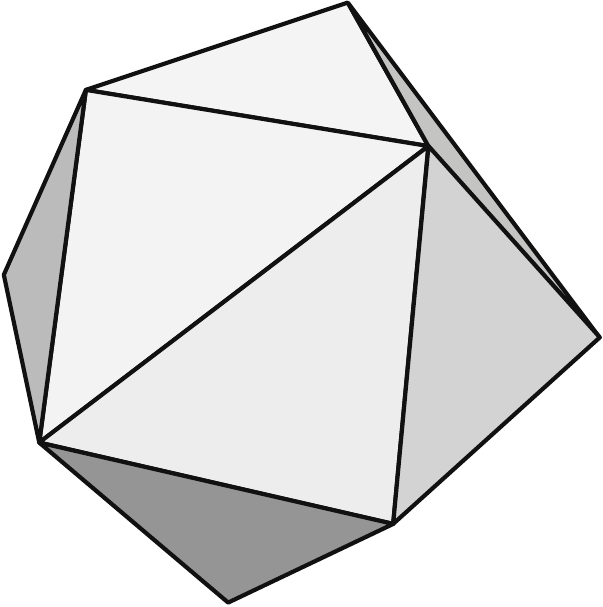} & 
\includegraphics[scale=0.3]{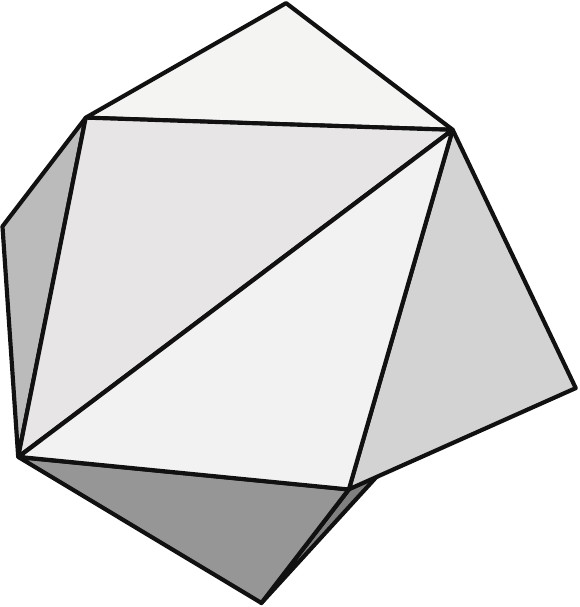} & 
\includegraphics[scale=0.3]{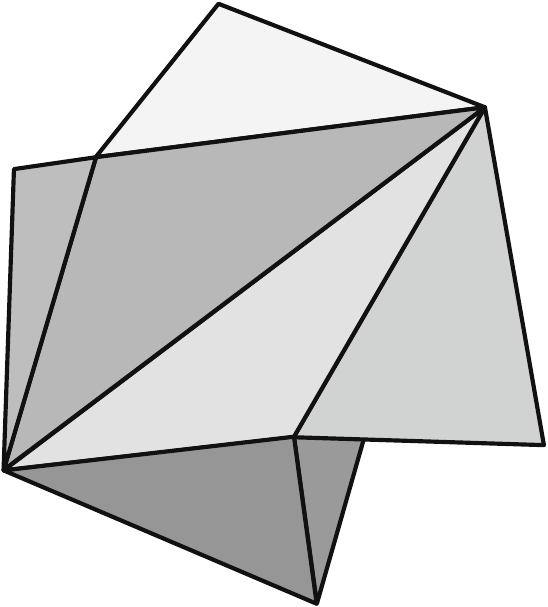} &
\includegraphics[scale=0.3]{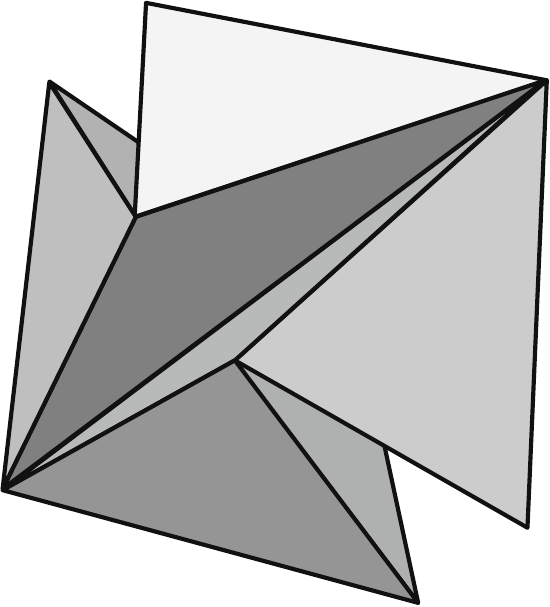} \\
\includegraphics[scale=0.3]{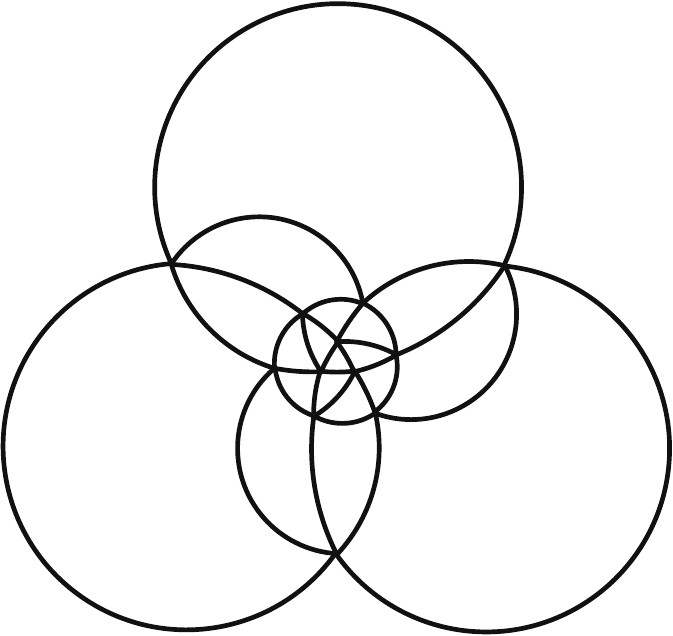} &
\includegraphics[scale=0.3]{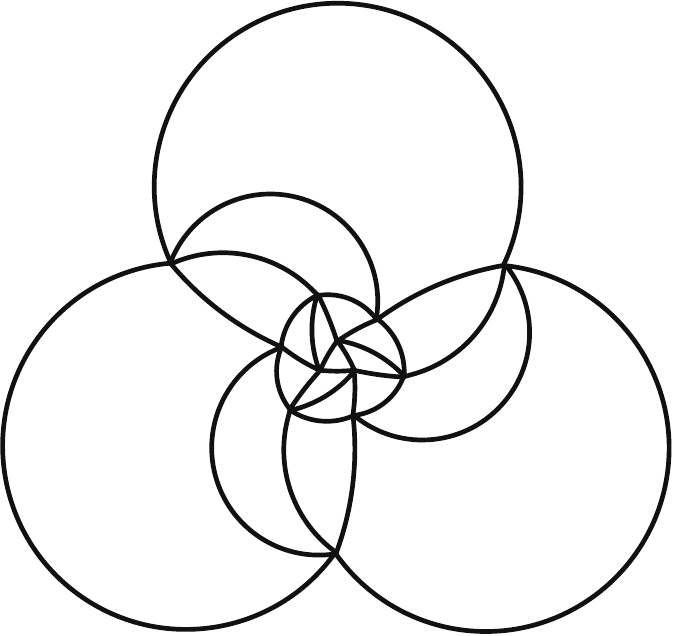} &
\includegraphics[scale=0.3]{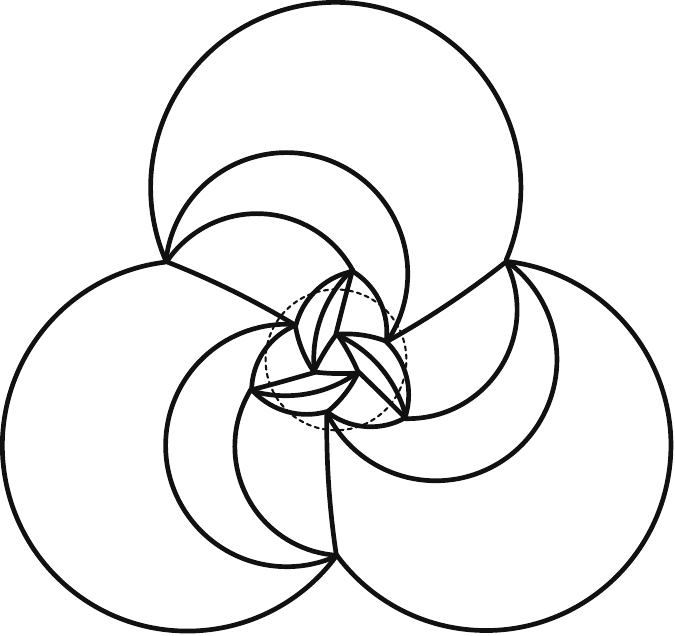} &
\includegraphics[scale=0.3]{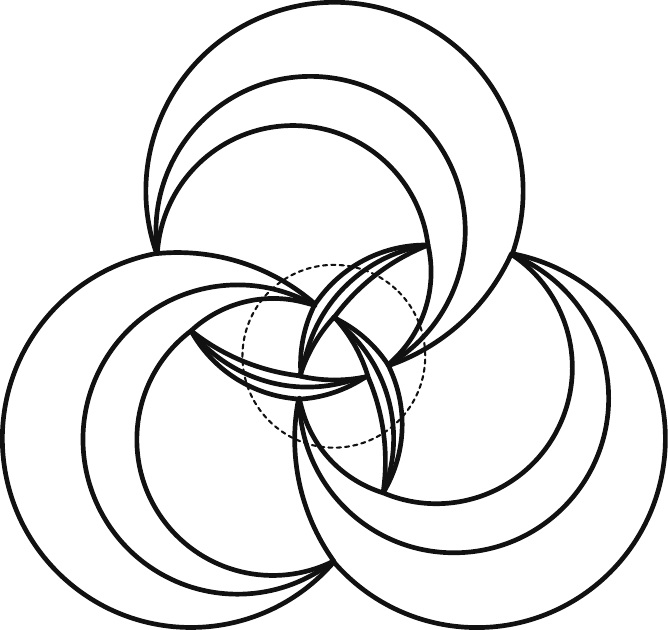} \\
 (a) & (b) & (c) & (d)
\end{tabular}
\caption{Four six-beaked shaddocks and stereographic projections of the corresponding spherical triangulations.  Depending on the twisting angle, these triangulations are (a) coherent, (b)~longitudinally shellable but not coherent, (c)~sinkable but not longitudinally shellable (Jessen's icosahedron), or (d) not sinkable.}
\label{F:shaddock}
\end{figure}

Both longitudinal shellability and sinkability depend on a specific choice of antipodal points as the north and south poles to define longitudes; neither property is invariant under rotations of the sphere.  For example, rotating any Schönhardt triangulation by $90^\circ$ around the $x$-axis results in a longitudinally shellable and therefore sinkable triangulation.

\subsection{Our Results}

We describe and study a promising approach to efficiently morph shortest-path triangulations on the sphere, extending the previous results of Awartani and Henderson \cite{ah-sgts-87} and Kobourov and Landis \cite{kl-mpgss-06}.  At a high level, we propose finding suitable rotations of the sphere that make the source and target triangulations both sinkable, moving vertices of the rotated triangulations into the southern hemisphere along longitudes, and finally reducing to a planar morphing problem.

We emphasize that we do \emph{not} develop this strategy into a complete algorithm; the existence of an efficient spherical morphing algorithm remains an open problem.  We also note that sinking is not \emph{required} to construct spherical morphs; for example, we can directly morph any Schönhardt or six-beaked-shaddock triangulation into any other by twisting the faces.

\medskip
We begin in Section \ref{S:background} by reviewing relevant definitions, describing our proposed morphing strategy in detail, and proving some preliminary results.  In particular, we prove that a triangulation $\Map$ is sinkable if and only if there is an isomorphic \emph{weak} triangulation $\SouthMap$ (intuitively, a triangulation in which some faces may be degenerate) whose vertices are on the same longitudes but below the equator.  We also show that any efficient algorithm for morphing spherical triangulations implies an efficient algorithm to morph shortest-path embeddings of 3-connected planar graphs on the sphere.

In Section~\ref{S:shelling}, we formally define the \emph{longitudinal shelling} condition introduced by Awartani and Henderson.  We provide several combinatorial characterizations, in terms of different directed versions of the triangulation or its dual graph, each of which can be tested in $O(n)$ time.  We also reiterate Awartani and Henderson's proof that every longitudinally shellable triangulations is sinkable.  Finally, by combining Awartani and Henderson's results with more recent planar morphing algorithms \cite{aabcd-hmpgd-17,k-cdhgl-21}, we describe an efficient algorithm to construct morphs in which each edge of each intermediate embedding has at most one \emph{bend}.

Rotating a longitudinally unshellable triangulation can make it longitudinally shellable.  In Section \ref{S:shellrot}, we present an algorithm that either finds a rotation of a given triangulation $\Map$ that is longitudinally shellable, or reports correctly that no such rotation exists, in $O(n^{5/2}\log^3 n)$ time. Our algorithm searches each of the $O(n^2)$ cells in the arrangement of great circles determined by the edges of $\Map$, using a dynamic reachability data structure of Diks and Sankowski~\cite{ds-tptc-07}.  We also construct a spherical triangulation that has no longitudinally shellable rotation.

In Section \ref{S:sink}, we describe an efficient algorithm to determine whether a given sphere triangulation is \emph{sinkable}.  First we prove that sinkability is equivalent to the feasibility of a linear program with $O(n)$ variables and $O(n)$ constraints.  By identifying the optimal basis for this linear program, we then show that it can be solved in $O(n^{\omega/2})$ time, assuming the underlying linear system is non-singular.

As a tool for building intuition, we implemented a suite of algorithms to construct  sphere triangulations, test for shellability and sinkability, construct Awartani--Henderson embeddings, and visualize longitudinal morphs.  In Section \ref{S:experiment}, we present results of some lightweight experiments with “ugly” spherical triangulations, which suggest that “in practice”, one can find a longitudinally shellable or sinkable rotation of a triangulation by testing only a small constant number of random rotations.

Finally, we conclude in Section \ref{S:outro} by presenting several directions for further research, including several conjectures suggested by our experimental results.

% ----------------------------------------------------------------------------------------
\section{Background and Preliminary Results}
\label{S:background}

\subsection{Cartography}

We consider drawings of graphs on the unit sphere $S^2 = \set{(x,y,z)\mid x^2+y^2+z^2=1}$.
The \EMPH{north pole} is the point $(0,0,1)$; the \EMPH{south pole} is the point $(0,0,-1)$; the \EMPH{equator} is the intersection of the unit sphere with the plane $z=0$.  A \EMPH{longitude} is an open great-circular arc whose endpoints are the north and south poles.  Every point except the poles lies on a unique longitude.

Throughout the paper, we fix an arbitrary simple maximal planar graph $G$ with $n\ge 4$ vertices, arbitrarily indexed from $1$ to $n$.  An \EMPH{embedding} of $G$ on the sphere is an injective map from $G$ to~$S^2$, or less formally, a drawing of $G$ where edges are arbitrary simple curves that intersect only at their shared endpoints.  Because $G$ is maximal planar, every embedding of $G$ on the sphere is a \EMPH{triangulation}, meaning every face is bounded by a cycle of three distinct edges.

Unless explicitly indicated otherwise, we consider only \EMPH{generic shortest-path} triangulations, meaning (1)~every edge is a great-circular arc with length less than $\pi$; (2)~no great circle contains more than two vertices, and (3) no great circle through the poles contains more than one vertex.  Condition (1) is the natural analogue on the sphere of straight-line embeddings in the plane.  Condition~(3) implies that vertices lie on distinct longitudes.

Classical theorems of Steinitz \cite{s-pr-1916,sr-vtp-34,g-gppg-07} and Whitney \cite{w-cgcg-32,b-sepwu-21} imply that~$G$ can be embedded as a generic shortest-path triangulation, and that this embedding is unique up to homeomorphism.  A \EMPH{face of $G$} is a face of this essentially unique embedding.  We sometimes identify each face as a triple $(i,j,k)$ of vertices (or their indices) in counterclockwise order around its interior.

Every generic shortest-path triangulation has a unique \EMPH{north face} that contains the north pole in its interior and a unique \EMPH{south face} that contains the south pole in its interior.  The remaining \EMPH{non-polar} faces are of two types: An \EMPH{up-face} has one vertex (called its \EMPH{apex}) that is directly north of its opposite edge (called its \EMPH{base}); symmetrically, a \EMPH{down-face} has one vertex (its apex) directly south of its opposite edge (its base).  The \EMPH{legs} of a non-polar face $f$ are the edges of $f$ incident to the apex of $f$.  The faces immediately north and south of any vertex~$i$ are respectively denoted \EMPH{$\AboveF{i}$} and \EMPH{$\BelowF{i}$}.  For each vertex $i$, $\AboveF{i}$ is either the north face or a down-face, and $\BelowF{i}$ is either the south face or an up-face.  See Figure \ref{F:updown}.

\begin{figure}[htb]
\centering
\includegraphics[scale=0.4]{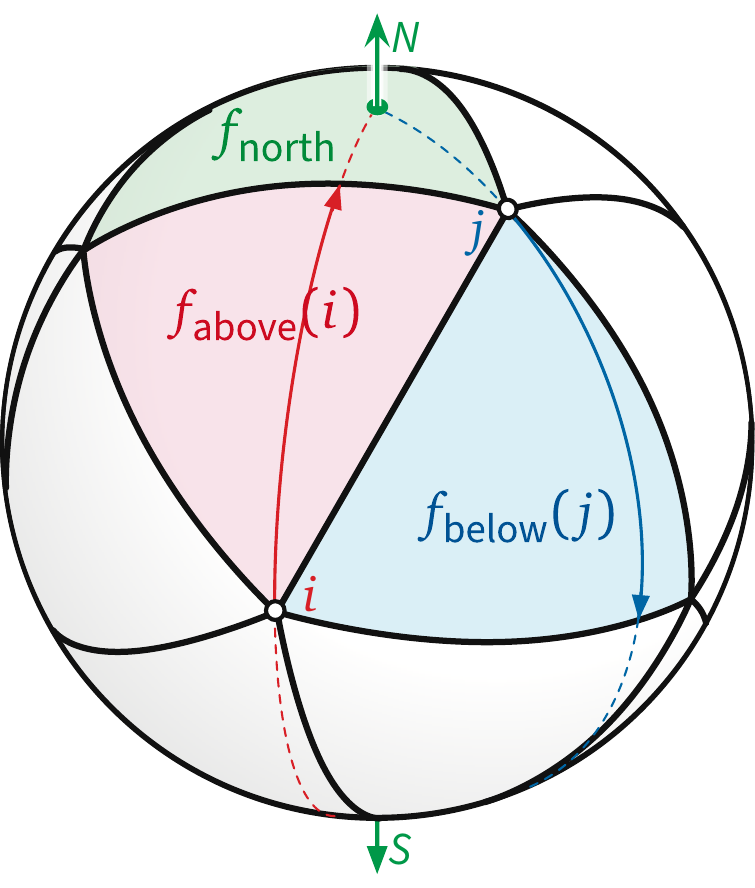}
\caption{An up-face $\AboveF{i}$, a down-face $\BelowF{j}$, and the north face $\NorthF = \AboveF{j}$ of a spherical triangulation.  Edge $ij$ is a leg of both $\AboveF{i}$ and $\BelowF{j}$.}
\label{F:updown}
\end{figure}

A face of a (generic) triangulation is \EMPH{everted} if it has area greater than $2\pi$, or equivalently, if it contains an open hemisphere, or equivalently, if it contains a pair of antipodal points.  Every triangulation contains at most one everted face.  We call a triangulation \EMPH{full} if no face is everted, and \EMPH{southern} if every vertex lies below the equator (which implies that the north face is everted).

A connected region $R$ on the sphere is \EMPH{$\theta$-monotone} if every longitude is either disjoint from~$R$ or has connected intersection with~$R$.

\subsection{Coordinates}
\label{SS:coords}

Throughout the paper, we represent spherical triangulations \emph{implicitly} as  projections of certain polyhedra onto the unit sphere.  Specifically, we represent each vertex using signed homogeneous coordinates \cite{s-opg-87,s-pcg-88,s-opgfg-91}---for any scalar $\lambda > 0$, the coordinate vectors $(x,y,z)$ and $(\lambda x, \lambda y, \lambda z)$ represent the same point on the  sphere.  We can freely rescale the vertices, each independently, without changing the underlying spherical triangulation; although it is never actually necessary, this rescaling is sometimes convenient for purposes of discussion and intuition.   For example, we can implicitly interpret any southern triangulation as a \emph{planar} triangulation by scaling vertices onto the tangent plane $z = -1$; this scaling is equivalent to the \emph{gnomonic} or \emph{central} projection $(x,y,z)\mapsto (-x/z, -y/z, -1)$.\footnote{Gnomonic projection should not be confused with \emph{stereographic} projection $(x,y,z) \mapsto (\frac{x}{1-z}, \frac{y}{1-z})$, which maps shortest paths on the sphere to circular arcs in the plane.  We use stereographic projections to visualize graphs on the sphere in Figures \ref{F:schonhardt}, \ref{F:shaddock}, and \ref{F:directed-graphs}, but they play no other role in our results.}  Awartani and Henderson \cite{ah-sgts-87} analyze longitudinal morphing (defined below) by implicitly scaling vertices onto the unit cylinder $x^2+y^2 = 1$.

Every full triangulation is the projection of a \emph{star-shaped} polyhedron whose visibility kernel contains the origin.  Every non-full triangulation is the projection of a polyhedron that is star-shaped except for one facet. 

\subsection{Longitudinal Morphing}

Two embeddings $\Map\colon G\to S^2$ and $\Map'\colon G\to S^2$ are \EMPH{isomorphic} if they define the same rotation system, or equivalently, if there is an orientation-preserving homeomorphism $H\colon S^2\to S^2$ such that $\Map' = h\circ \Map$.  Two isomorphic embeddings are \EMPH{$\theta$-equivalent} if each vertex of $G$ lies on the same longitude in both embeddings.

A \EMPH{homotopy} between two embeddings $\Map\colon G\to S^2$ and $\Map'\colon G\to S^2$ is a continuous function $H\colon [0,1]\times G\to S^2$ such that $H(0,\cdot) = \Map$ and $H(1,\cdot) = \Map'$.  Any pair of isomorphic embeddings~$\Map$ and $\Map'$ are connected by an \EMPH{isotopy}, which is a homotopy $H$ such that every intermediate drawing $\Map_t = H(t, \cdot)$ is an embedding.  (The edges of these intermediate embeddings can be arbitrary simple curves, not necessarily shortest paths.)  A~\EMPH{morph} is an isotopy in which the edges of every intermediate embedding~$\Map_t$ are shortest paths.  Finally, a \EMPH{longitudinal morph} is a morph in which vertices move only along their longitudes.

\begin{lemma}
\label{L:shell-sink}
Any two $\theta$-equivalent triangulations are connected by a longitudinal morph.
\end{lemma}

\begin{proof}
Let $\Map$ and $\Map'$ be two $\theta$-equivalent triangulations.  By scaling coordinate vectors as described in Section \ref{SS:coords} if necessary, we can assume that each vertex has the same $x$- and $y$-coordinates in both embeddings.  Let $(x_i, y_i, z_i)$ and $(x_i, y_i, z'_i)$ denote the coordinate vectors of each vertex $i$ in~$\Map$ and $\Map'$, respectively.

We define a longitudinal morph by linearly interpolating the $z$-coordinate of each vertex from~$z_i$ to $z'_i$.  For each $t\in [0,1]$, let $\Map_t$ denote the shortest-path \emph{drawing} where each vertex~$i$ has coordinates $(x_i, y_i, (1-t)z_i + t\cdot z'_i)$.  To complete the proof, it remains only to show that each drawing $\Map_t$ is an embedding, or equivalently, that no face collapses to an arc during the morph.

For any three vertices $i,j,k$ and any $t\in[0,1]$, let $\vol_t(i,j,k)$ denote the determinant
\[
	\det\begin{pmatrix}
		x_i & y_i & (1-t)\cdot z_i + t\cdot z'_i \\[0.5ex]
		x_j & y_j & (1-t)\cdot z_j + t\cdot z'_j \\[0.5ex]
		x_k & y_k & (1-t)\cdot z_k + t\cdot z'_k
	\end{pmatrix}
\]
This volume determinant is a linear function of $t$; specifically,
\[
	\vol_t(i,j,k) ~=~ (1-t)\cdot \vol_0(i,j,k) ~+~ t\cdot \vol_1(i,j,k).
\]
Suppose $i,j,k$ are the vertices of a \emph{non-polar} face of $\Map$ in counterclockwise order, so that $\vol_0(i,j,k) > 0$.  Because $\Map$ and $\Map'$ are isomorphic, we have $\vol_1(i,j,k) > 0$; linearity immediately implies that $\vol_t(i,j,k) > 0$ for all $t\in[0,1]$.

Finally, for all $t$, the north pole lies in the interior of the north face of $\Map_t$, and the south pole lies in the interior of the south face of $\Map_t$.  Thus, these faces may become (un)everted, but they always have non-empty interiors.
\end{proof}

\begin{corollary}
\label{C:sink-south}
A triangulation $\Map$ is sinkable if and only if there is a southern triangulation that is $\theta$-equivalent to $\Map$.
\end{corollary}

We extend this argument slightly by considering \EMPH{weak triangulations}, which are drawings of~$G$ that may include degenerate faces, but no inverted faces or crossing edges.  (Weak triangulations are special cases of \emph{weak embeddings} studied by Akitaya, Fulek, and Tóth \cite{aft-rweg-18} and Fulek and Kyncl \cite{fk-htamg-18}.)  A~weak triangulation~$\Map'$ is $\theta$-equivalent to a triangulation $\Map$ if vertices lie on the same longitudes in~$\Map$ and $\Map'$ and every non-degenerate non-polar face of $\Map'$ has the same orientation as the corresponding face of $\Map$.

%\textcolor{OliveGreen}{The following lemma is immediate:}
%
%\begin{lemma}
%\label{L:weak-sink}
%Any triangulation $\Map_0$ and any weak triangulation $\Map_1$ that is longitudinally equivalent %to~$\Map_0$ are connected by a homotopy $H$, such that for any $0<\e<1$, the restriction of $H$ %to the interval $[0,1-\e]$ is a longitudinal morph.
%\end{lemma}
%\unskip
%\color{BrickRed}
\begin{corollary}
\label{C:weak-sink}
A triangulation $\Map$ is sinkable if any only if there is a southern weak triangulation that is $\theta$-equivalent to~$\Map$.
\end{corollary}

\begin{proof}
Fix a triangulation $\Map$ and a $\theta$-equivalent southern weak triangulation $\Map'$.  For each $t\in [0,1]$, let $\Map_t$ denote the shortest-path drawing where each vertex~$i$ has coordinates $(x_i, y_i, z_i(t))$, where $z_i(t) = (1-t)z_i + t\cdot z'_i$.  The proof of Lemma \ref{L:shell-sink} implies that $\Map_t$ is an embedding for all $0\le t < 1$.

Let $\tau$ be the smallest value such that $z_i(\tau) \le 0$ for all $i$.  Because $z_i(1) = z'_i < 0$ for all $i$, we must have $\tau<1$, and therefore $z_i((\tau+1)/2) < 0$ for all $i$.  We conclude that $\Map_{(\tau+1)/2}$ is a southern triangulation (with no degenerate faces) that is $\theta$-equivalent to $\Map_0$, and the retriction of $H$ to $[0,(\tau+1)/2]$ is a longitudinal morph from $\Map$ to $\Map_{(\tau+1)/2}$.
\end{proof}

\subsection{Morphing Strategy}
\label{SS:strategy}

In this section, we propose a strategy for morphing between isomorphic shortest-path triangulations on the sphere that slightly generalizes strategies previously used by Awartani and Henderson \cite{ah-sgts-87} and Kobourov and Landis \cite{kl-mpgss-06}.  Recall that a triangulation is \EMPH{sinkable} if it can be longitudinally morphed to a $\theta$-equivalent triangulation with all vertices below the equator. A \EMPH{rotation} of a spherical triangulation $\Gamma$ is the triangulation $\rho\circ\Gamma$ for some rigid motion $\rho\colon S^2\to S^2$ of the sphere.\footnote{Rotations should not be confused with \emph{rotation systems}, which encode the cyclic order of edges around each vertex in a planar or spherical embedding.}  Our strategy rests on the following (admittedly optimistic) conjecture:

\begin{conjecture}
\label{C:sinkrot}
Every shortest-path triangulation $\Gamma$ of the sphere can be rotated to obtain a sinkable triangulation.
\end{conjecture}

Conjecture \ref{C:sinkrot} implies an efficient algorithm to morph any triangulation $\Map_0$ into any isomorphic target triangulation $\Map_1$ through an intermediate coherent triangulation $\Map_{1/2}$.  Steinitz's theorem \cite{s-pr-1916,sr-vtp-34,g-gppg-07} implies that an isomorphic coherent triangulation $\Map_{1/2}$ exists; moreover, we can compute such a triangulation in $O(n)$ time using an algorithm of Das and Goodrich \cite{dg-cop3d-97}.  Given any coherent triangulation, we can project its underlying convex polyhedron from a point on the positive $z$-axis just outside the polyhedron to the plane $z=-1$, and then centrally project the resulting planar triangulation back to the southern hemisphere, to obtain a $\theta$-equivalent \emph{southern} triangulation.  It follows that \emph{every} rotation of $\Gamma_{1/2}$ is sinkable.

We construct a morph from $\Map_0$ to $\Map_{1/2}$ in several stages as follows.  First, we rotate $\Map_0$ so that the resulting triangulation $\Map'_0$ is sinkable (as guaranteed by Conjecture \ref{C:sinkrot}) and then longitudinally morph $\Map'_0$ to a southern triangulation $\Map''_0$.\footnote{If the triangulation $\Map_0$ is not full, we can rotate it directly to a southern triangulation $\Map'_0 = \Map''_0$.}  Similarly, we rotate $\Map_{1/2}$ so that the resulting triangulation $\Map'_{1/2}$ has the same north face as $\Map'_0$ and $\Map''_0$, and then longitudinally morph $\Map'_{1/2}$ to a southern triangulation $\Map''_{1/2}$.  Projection from the center of the sphere maps~$\Map''_0$ and $\Map''_{1/2}$ to isomorphic straight-line triangulations in the plane $z=-1$, each with the same triangular outer face.  We can construct a morph between these planar triangulations using any of several algorithms \cite{c-dprc-44, h-chps1-73, bs-li-78, fg-mti-99, aabcd-hmpgd-17, el-ptmme-23, k-cdhgl-21}; lifting this planar morph back to the sphere yields a spherical morph from~$\Map''_0$ to $\Map''_{1/2}$.  Our final morph is the concatenation of the rotation $\Map_0 \leadsto \Map'_0$, the longitudinal morph $\Map'_0 \leadsto \Map''_0$, the lifted planar morph $\Map''_0 \leadsto \Map''_{1/2}$, the reverse of the longitudinal morph $\Map'_{1/2} \leadsto \Map''_{1/2}$, and the reverse of the rotation $\Map_{1/2} \leadsto \Map'_{1/2}$.  See Figure~\ref{F:morph-example} for an example (without the initial and final rotations).

\begin{figure}[htb]
\centering\sffamily
\begin{tabular}{c@{}c@{}c@{}c@{}c}
	\raisebox{-0.5\height}{\includegraphics[scale=0.3]{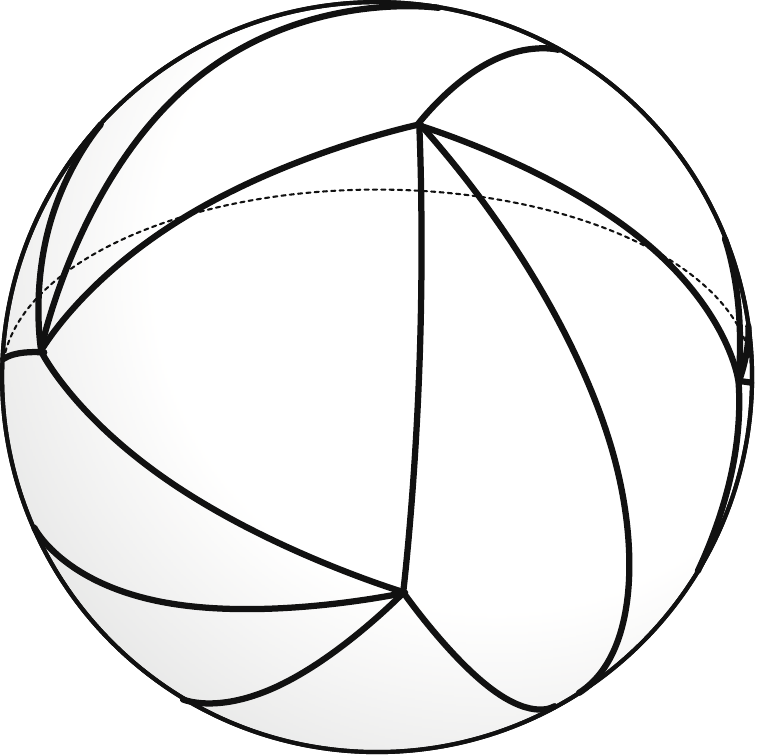}} &
	$\xrightarrow{\text{~sink~}}$ &
	\raisebox{-0.5\height}{\includegraphics[scale=0.3]{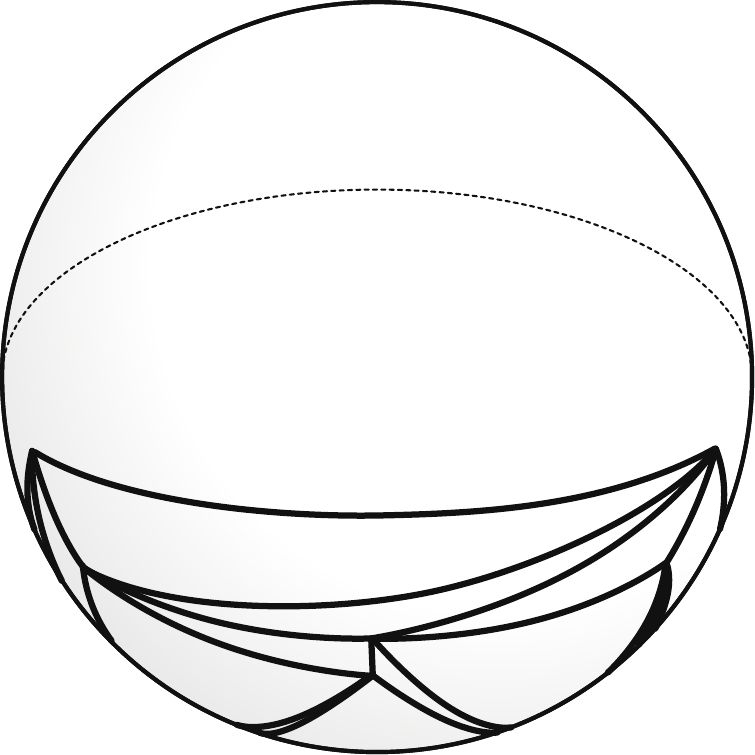}} &
	$\xrightarrow{\text{project}}$ &
	\hspace{-2em}
	\raisebox{-0.5\height}{\includegraphics[scale=0.3]{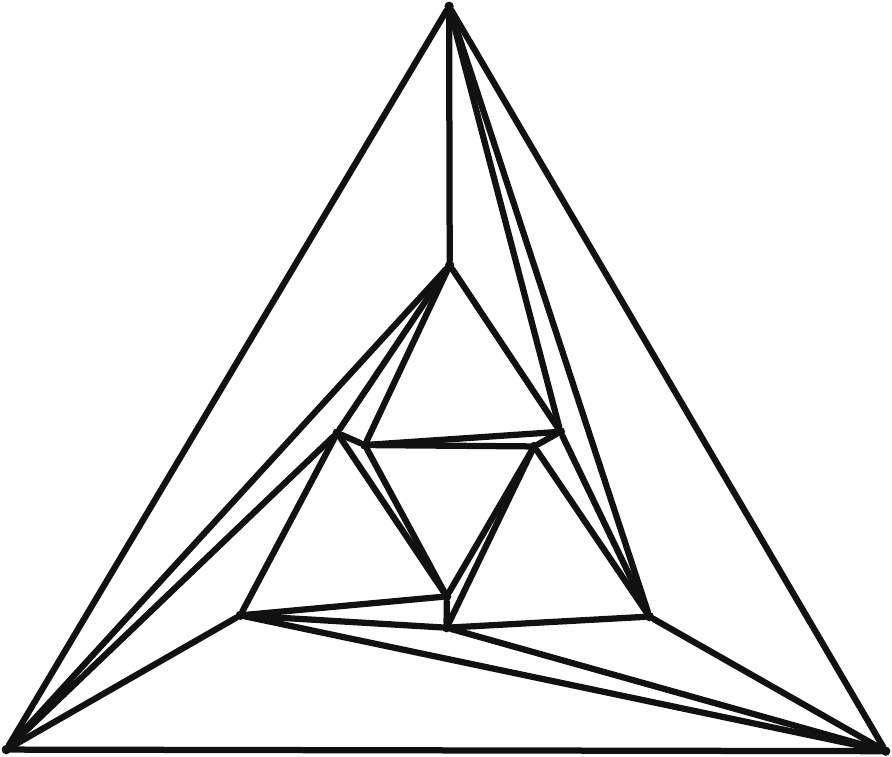}}
	\\
	&&
	$\Big\downarrow \rlap{$\scriptstyle \text{lifted planar morph}$}$
	&&
	\hspace{-2em}
	$\Big\downarrow \rlap{$\scriptstyle \text{planar morph}$}$
	\\
	\raisebox{-0.5\height}{\includegraphics[scale=0.3]{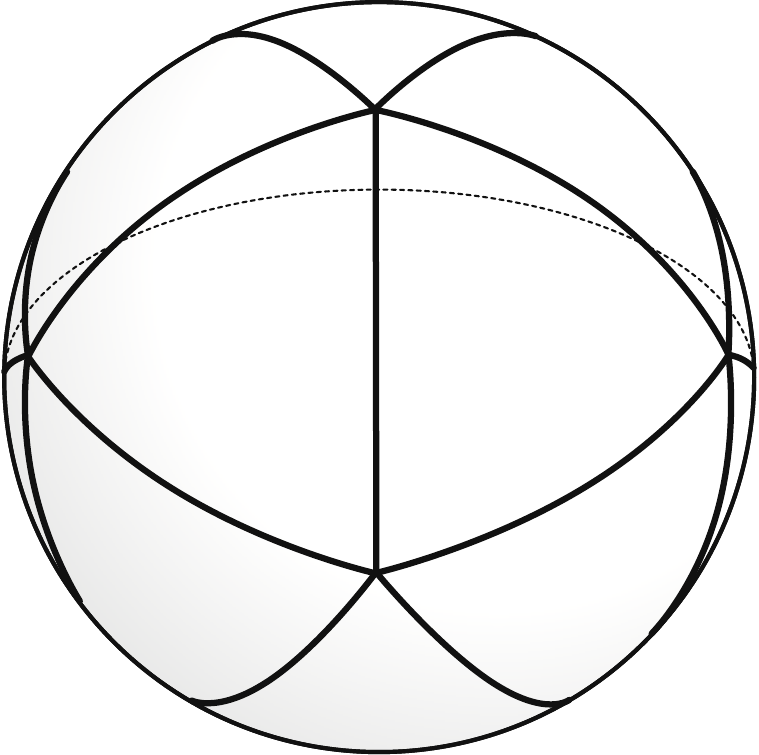}} &
	$\xleftarrow{\text{unsink}}$ &
	\raisebox{-0.5\height}{\includegraphics[scale=0.3]{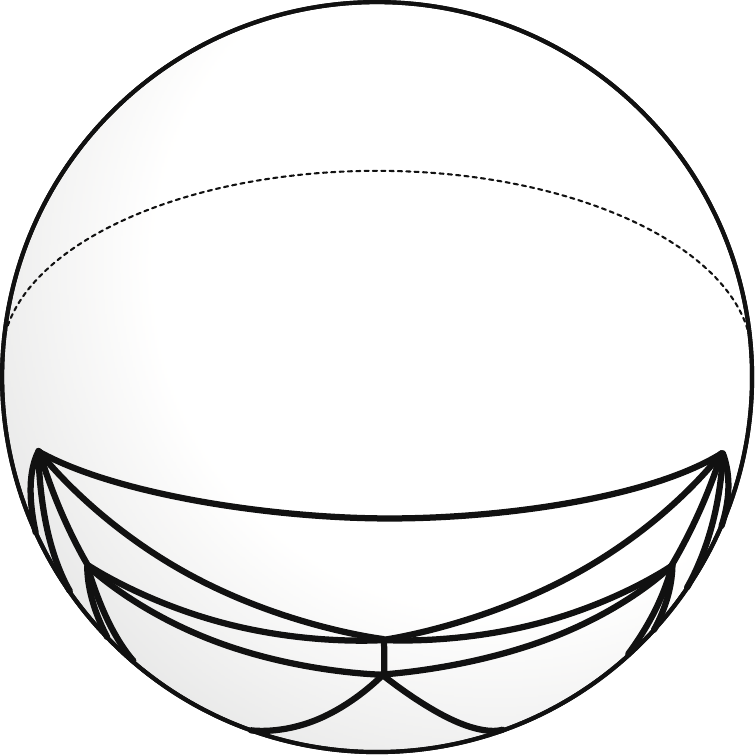}} &
	$\xleftarrow{\text{~~lift~~}}$ &
	\hspace{-2em}
	\raisebox{-0.5\height}{\includegraphics[scale=0.3]{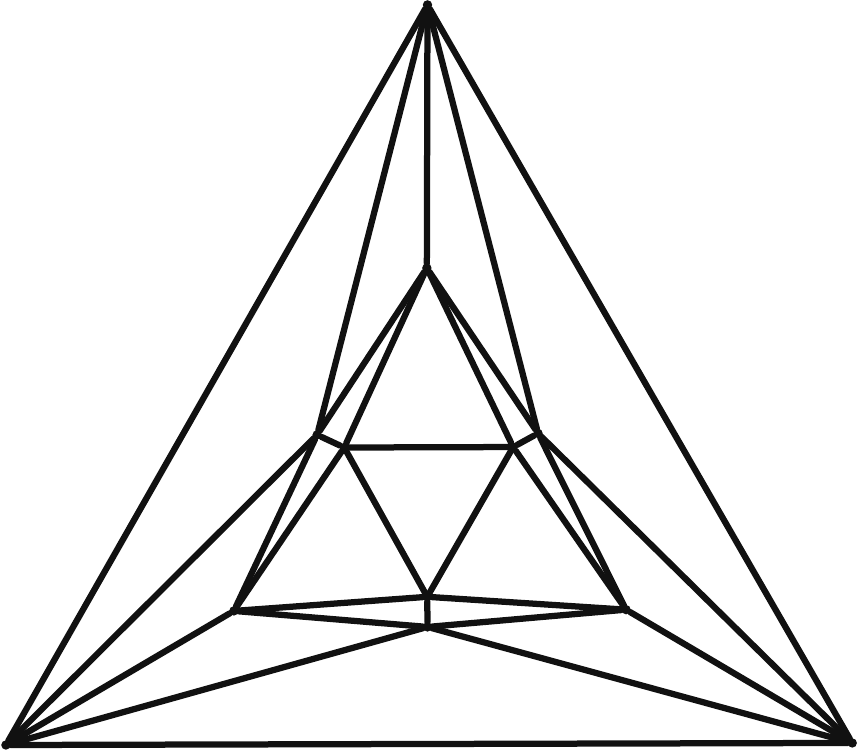}}
\end{tabular}
\caption{Morphing a six-beaked shaddock triangulation \cite{j-oi-67, d-sb-71} into a regular icosahedral triangulation; compare with Kobourov and Landis \cite[Figure 1]{kl-mpgss-06}.}
\label{F:morph-example}
\end{figure}

We construct a morph $\Map_1 \leadsto \Map_{1/2}$ using the same strategy.  The final morph  $\Map_0 \leadsto \Map_1$ is the concatenation of $\Map_0 \leadsto \Map_{1/2}$ and the reverse of $\Map_1 \leadsto \Map_{1/2}$.

\subsection{Morphing 3-Connected Embeddings}

To close this section, we describe how to extend any efficient morphing algorithm for  spherical triangulations to shortest-path embeddings of 3-connected graphs.

Our reduction relies on an algorithm to triangulate the faces of a 3-connected embedding with additional shortest paths.  We can triangulate any face that fits in a hemisphere by projecting it to a simple polygon to the plane~\cite{bde-ttp-96}, but triangulating faces that do not fit in a hemisphere requires more specialized algorithms \cite{o-cgc51-08,bk-sapsr-82}, some of which have never been published~\cite{h-esptu-08}.  See O'Rourke \cite{o-cgc51-08} for a thorough discussion of subtleties in triangulating spherical polygons.

\begin{lemma}
\label{L:triangulate}
Let $\Map$ be a shortest-path embedding of a 3-connected planar graph on the sphere.  We can triangulate the faces of $\Map$, using at most one Steiner point, in $O(n\log n)$ time.
\end{lemma}

\begin{proof}[sketch]
We adapt an algorithm of Fournier and Montuno \cite{fm-tspep-84} for triangulating planar polygons, which is also described by O'Rourke~\cite[Chapter~2]{o-cg-98}, using techniques developed by Halperin~\cite{h-esptu-08}.  We provide only a sketch of our algorithm here, omitting several straightforward proof and implementation details.

\def\Trap{\mathcal{L}(\Map)}

If necessary, we rotate~$\Map$ so that the north and south poles are in different faces and, if any face of $\Map$ contains a hemisphere, that face also contains the north pole.  We define the \emph{longitudinal decomposition} $\Trap$ of~$\Map$ by extending longitudinal segments north and south from each vertex to the next intersection with~$\Map$, or to one of the poles if there is no such intersection.  Each face of $\Trap$ is a (possibly degenerate) spherical trapezoid that is contained in a hemisphere.  We can construct $\Trap$ in $O(n\log n)$ time, either by sweeping the sphere with a moving longitude \cite{sh-gip-76, bo-arcgi-79} or using a randomized incremental algorithm \cite{cs-arscg-89,m-fppa-90}.

The rest of the triangulation algorithm requires only $O(n)$ time.  We call a face $f$ of $\Trap$ \emph{boring} if two vertices of $\Map$ lie on the boundary of that face but not on the same edge of $\Trap$; we call the shortest path between those two vertices a \emph{boring diagonal}.  Inserting all boring diagonals into the original embedding $\Map$ partitions its faces into $\theta$-monotone pieces.  Two of these $\theta$-monotone pieces contain the poles; each of the others is a \emph{monotone mountain}, whose boundary consists of a $\theta$-monotone polygonal chain and a single shortest path.  Each monotone mountain can be triangulated in $O(n)$ time by repeatedly removing convex ears~\cite{fm-tspep-84}.  We can similarly remove convex ears from each polar piece until either the piece is fully triangulated or there is one remaining fragment that has no convex vertices.  We can triangulate the final “anti-convex” fragment by extending shortest paths from the north pole to each vertex.
\end{proof}

\begin{theorem}
\label{Th:3con}
Suppose we can construct a morph between two isomorphic shortest-path triangulations of the sphere in $O(T(n))$ time.  Then we can construct a morph between two isomorphic embeddings of the same 3-connected planar graph in $O(T(n) + n\log n)$ time.
\end{theorem}

\begin{proof}
Let $\Map_0$ and $\Map_1$ be any two isomorphic shortest-path embeddings of the same 3-connected planar graph on the sphere.  As in the previous section, we compute an isomorphic coherent triangulation $\Map_{1/2}$ in $O(n)$ time~\cite{dg-cop3d-97} and then construct a morph from $\Map_0$ to $\Map_1$ through this intermediate coherent triangulation.

We construct a morph from $\Map_0$ to $\Map_{1/2}$ as follows.  First, we triangulate the faces of $\Map_0$ in $O(n\log n)$ time, as described by Lemma \ref{L:triangulate}.  Then we triangulate each face of $\Map_{1/2}$ isomorphically to the corresponding face of $\Map_0$; this is always possible because the faces of $\Map_{1/2}$ are convex.  In particular, if the north face of $\Map_0$ contains a hemisphere, we introduce one new vertex inside the north face of $\Map_{1/2}$.  Finally, we construct a morph from the triangulation of $\Map_0$ to the triangulation of $\Map_{1/2}$, in at most $T(n+1) = O(T(n))$ time, and simply ignore the vertices and edges that are not part of the original embeddings.

We similarly construct a morph $\Map_1 \leadsto \Map_{1/2}$, and finally assemble a morph from $\Map_0$ to $\Map_1$ by concatenating $\Map_0 \leadsto \Map_{1/2}$ and the reverse of $\Map_1 \leadsto \Map_{1/2}$.
\end{proof}

We conjecture that similar reductions can be developed for spherical embeddings that are not 3-connected, mirroring reductions from morphing arbitrary planar straight-line graphs to morphing planar triangulations~\cite{aabcd-hmpgd-17,t-dpg-83,sg-cmcpt-01,sg-imct-03,gs-gipm-01}.  We leave the development of these more general reductions to future work.

% ------------------------------------------------------------------------------
\section{Longitudinal Shelling}
\label{S:shelling}

A \EMPH{shelling} of a spherical triangulation $\Map$ is an ordering $f_1, f_2,\allowbreak \dots,\allowbreak f_{2n-4}$  of the faces of $\Map$ such that the union of any prefix $\FaceUnion_k = \bigcup_{i\le k} f_i$ is either empty ($k=0$), the entire sphere ($k=2n-4$), or a topological disk.  We call a shelling \EMPH{longitudinal} if every disk $\FaceUnion_k$ is $\theta$-\nobreak monotone.  The triangulation $\Map$ is \EMPH{longitudinally shellable} if it has a longitudinal shelling.

Bruggesser and Mani's line shelling \cite{bm-sdcs-71} implies that every \emph{coherent} triangulation is longitudinally shellable.  Awartani and Henderson \cite{ah-sgts-87} proved that a triangulation is longitudinally shellable if it has a \emph{longitudinal seam}, that is, a longitude $\ell$ such that for any edge $e$, the intersection $\ell\cap e$ is either empty, an endpoint of $e$, or the entire edge $e$.

\subsection{Characterization}

We characterize the longitudinal shellability of $\Gamma$ in terms of three directed graphs, which are illustrated in Figure \ref{F:directed-graphs}.

\begin{itemize}
\item
The directed dual graph \EMPH{$\BDn{\Map}$} contains a directed edge from face $f_i$ to face $f_j$ if and only if~$f_i$ and $f_j$ share an edge in $\Gamma$, and some point in $f_i$ is due north (on the same longitude) of some point in $f_j$. %
The north face of $\Map$ has out-degree $3$ in $\Dn\Map$; each down-face has in-degree~$1$ and out-degree $2$; each up-face has in-degree~$2$ and out-degree~$1$, and the south face has in-degree $3$.

\item
The oriented primal graph \EMPH{$\BRt{\Gamma}$} is the directed dual of $\Dn{\Map}$.  For every undirected edge $ij$ of~$\Map$, the graph $\Rt{\Gamma}$ contains the directed edge $i\arcto j$ if and only if vertex $j$ is east of vertex~$i$, or more formally, if the determinant
\[
	\det\begin{pmatrix}
		0 & 0 & 1 \\
		x_i & y_i & z_i \\
		x_j & y_j & z_j \\
	\end{pmatrix}
	=
	\det\begin{pmatrix}
		x_i & y_i \\
		x_j & y_j \\
	\end{pmatrix}
\]
is positive, where $(x_i, y_i, z_i)$ is the coordinate vector for vertex $i$.  

\item
Finally, \EMPH{$\BDnVee{\Map}$} is a directed proper subgraph of $\Map$ whose edges are the legs of each down-face, each directed toward its apex.
\end{itemize}\unskip

\begin{figure}[htb]
\centering
\begin{tabular}{c@{\qquad}c}
	\includegraphics[scale=0.4]{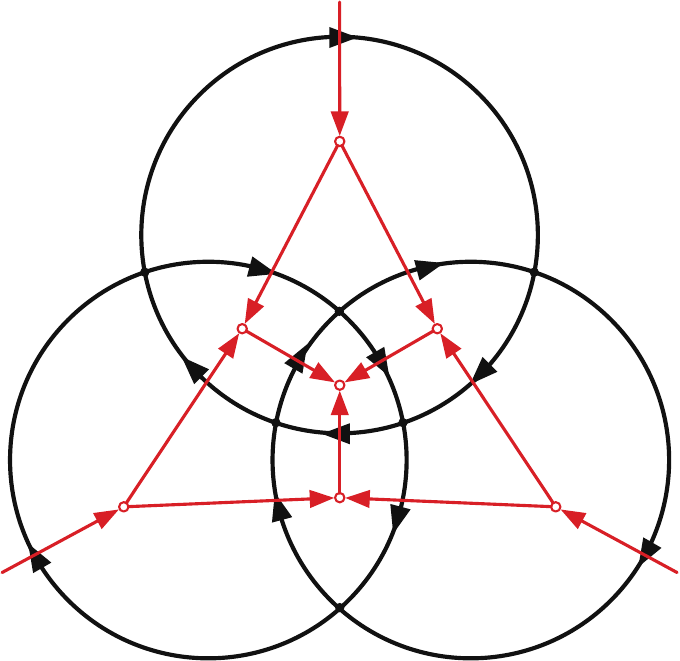} &
	\includegraphics[scale=0.4]{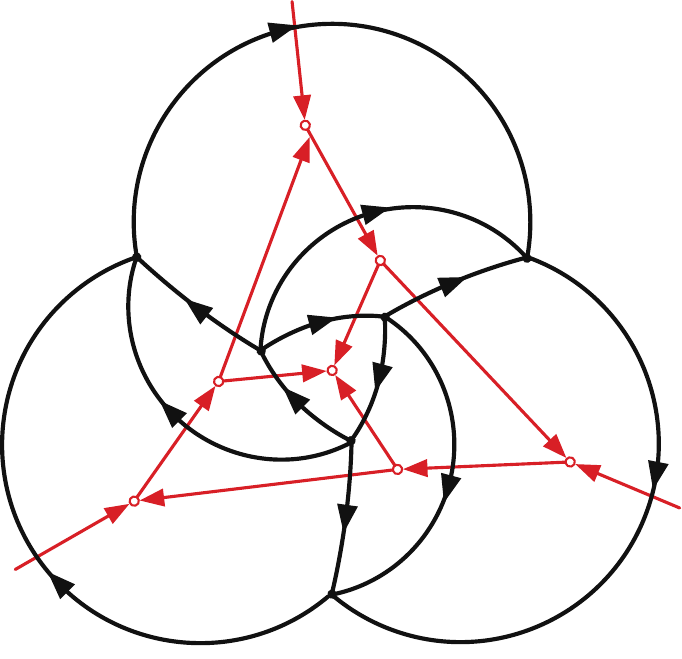}
	\\
	\includegraphics[scale=0.4]{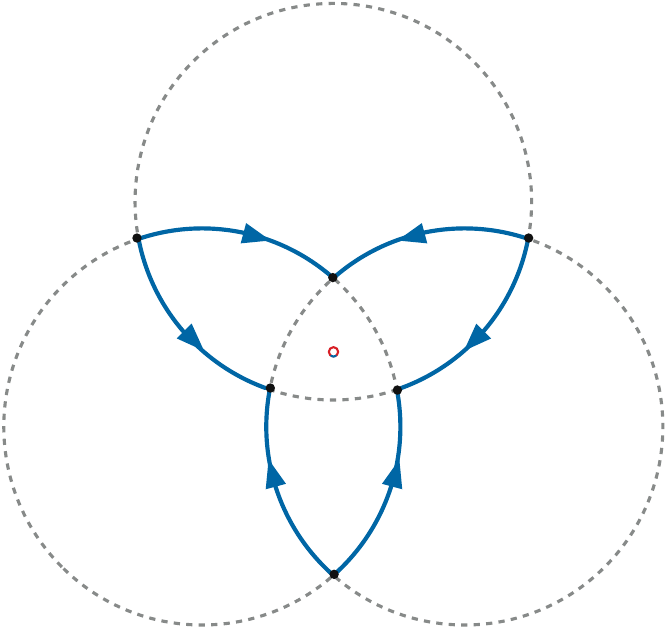} &
	\includegraphics[scale=0.4]{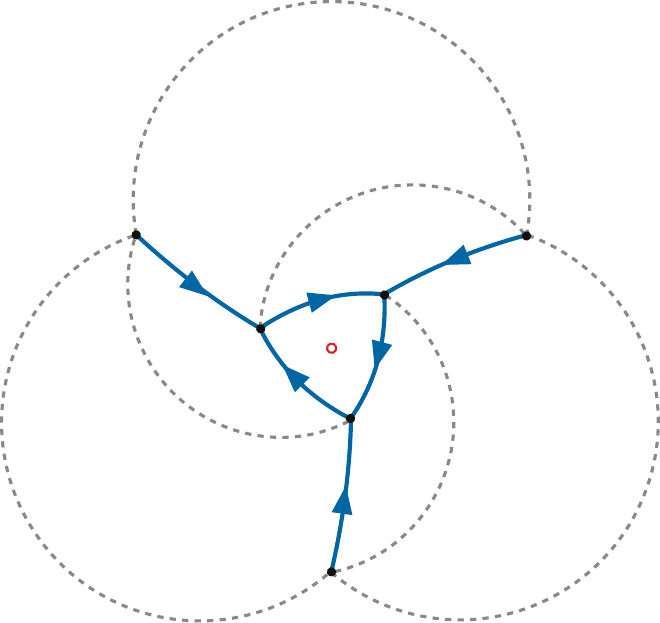}
\end{tabular}
\caption{Stereographic projections of directed graphs $\Dn\Map$ (red straight edges), $\Rt\Map$ (black curved edges), and $\DnVee\Map$ (blue, second row) for the Platonic octahedral triangulation (which is longitudinally shellable) and a Schönhardt triangulation  (which is not).  Compare with Figure~\ref{F:schonhardt}.}
\label{F:directed-graphs}
\end{figure}

\begin{lemma}
\label{L:shellable}
For every generic shortest-path triangulation $\Map$, the following conditions are equivalent:
\begin{enumerate}[(a)]\cramped
\item
$\Map$ is longitudinally shellable.
\item
$\Dn{\Map}$ is acyclic.
\item
$\Rt{\Map}$ is strongly connected.
\item
$\Rt{\Map}$ contains directed paths from $\NorthF$ to $\SouthF$ and from $\SouthF$ to $\NorthF$.
\item
$\DnVee{\Map}$ is acyclic.
%\item
%$\RtVee{\Map}$ is strongly connected.
\end{enumerate}
\end{lemma}

\begin{proof}
To prove (a)$\Rightarrow$(b), let $f_1, f_2, \dots, f_{2n-4}$ be a longitudinal shelling of $\Gamma$, and consider any two faces $f_k$ and $f_l$ that share a common edge, where $k<l$.  The disk $\FaceUnion_k$ is $\theta$-monotone, and therefore includes every point due north of $f_k$, but does not include $f_l$.  It follows that $f_l$ is south of~$f_k$, and thus $\Dn\Map$ contains the directed edge $f_k\arcto f_l$.  We conclude that $\Dn\Map$ is acyclic.

On the other hand, to prove (b)$\Rightarrow$(a), suppose $\Dn\Map$ is acyclic.  Let $f_1, f_2, \dots, f_{2n-4}$ be  any topological ordering of $\Dn\Map$.  
The north face is the only vertex of $\Dn\Map$ with in-degree $0$, so it must be $f_1$.
Consider the prefix union $\FaceUnion_k = \bigcup_{i\le k} f_i$ for some index $k$, and suppose some longitude~$\ell$ intersects two faces $f_i$ and $f_j$ such that $i\le k<j$.  The intersection $\ell\cap f_i$ must be north of $\ell\cap f_j$, since otherwise  $\Dn\Map$ would contain a path from $f_j$ to $f_i$, contradicting the topological order.  Thus, $\FaceUnion_k$ is a $\theta$-monotone disk.  It follows that $f_1, f_2, \dots, f_{2n-4}$ is a longitudinal shelling of $\Map$.

The equivalence (b)$\Leftrightarrow$(c) follows from the duality between directed cycles and directed edge-cuts in directed planar graphs \cite{k-lpapd1-93,n-atcop-01}.

The implication (c)$\Rightarrow$(d) is trivial.  To prove the converse (c)$\Leftarrow$(d), suppose $\Rt{\Map}$ contains directed paths $\pi_\downarrow$ from $\NorthF$ to $\SouthF$ and $\pi_\uparrow$ from  from $\SouthF$ to $\NorthF$.  Consider an arbitrary vertex $i$.  Because every edge in $\Rt{\Map}$ is oriented from west to east, following edges of $\Rt{\Map}$ starting at $i$ must eventually reach some vertex of $\pi_\downarrow$.  Thus, $\Rt{\Map}$ contains a directed walk from $i$ to $\SouthF$.  Symmetrically, $\Rt{\Gamma}$ contains a directed walk from $\SouthF$ to $i$, starting with some prefix of $\pi_\uparrow$.  We conclude that~$\Rt{\Map}$ is strongly connected.

The implication (b)$\Rightarrow$(e) is straightforward.  If $i\arcto j$ is an edge of $\DnVee{\Map}$, then $\AboveF{i}$ is north of $\AboveF{j}$, so there is a directed path from $\AboveF{i}$ to $\AboveF{j}$ in $\Dn\Map$.  Thus, if $\DnVee{\Map}$ contains a directed cycle, then $\Dn\Map$ also contains a directed cycle.

\medskip
It remains only to prove the implication (e)$\Rightarrow$(b).  Suppose $\Dn{\Map}$ contains a directed cycle.  Let~$C$ be any strongly connected component of~$\Dn{\Map}$ with more than one vertex, and let~$\Sleeve$ denote the union (or “sleeve”) of faces of $\Map$ whose dual vertices are in~$C$.  The north and south faces are respectively the unique source and sink vertices of $\Dn{\Map}$, so they cannot be contained in $C$ or in $\Sleeve$.  The following claims imply that $\Sleeve$ is a $\theta$-monotone annulus that separates the polar faces of $\Gamma$.

\begin{claim}
$\Sleeve$ intersects every longitude.
\end{claim}

\begin{subproof}
We prove this claim by adapting an argument of Palazzi and Snoeyink~\cite{ps-crrbs-94}.  For any two faces $f$ and $f'$ in $C$, let $f\Downarrow f'$ indicate that some longitude $\ell$ intersects both $f$ and $f'$, and $f\cap \ell$ is north of $f'\cap \ell$.  Let $f_1 \Downarrow f_2 \Downarrow \cdots \Downarrow f_r\Downarrow f_1$ be the minimum-length cycle through the faces in $C$ with respect to this relation.  We trivially have $r\ge 3$, since no face can be both north and south of another face.  If any longitude does not intersect $\Sleeve$, then  some face $f_i$ in this cycle is furthest east, which implies $f_{i-1} \Downarrow f_{i+1}$, contradicting the minimality of the cycle.
\end{subproof}

\begin{claim}
$\Sleeve$ is $\theta$-monotone.
\end{claim}

\begin{subproof}
Consider any longitude $\ell$, and let $a$ and $b$ respectively denote the northernmost and southernmost point in $\ell\cap\Sleeve$.  Let $f'_a, f'_{a+1}, \dots, f'_b$ denote the sequence of faces of~$\Map$ that intersect the longitudinal arc from $a$ to $b$, in order from north to south; in particular, $f'_a$ is the face of $\Map$ just south of~$a$, and $f'_b$ is the face of $\Map$ just north of~$b$.  We easily observe that $f'_a\in C$ and~$f'_b\in C$, and that $\Dn{\Map}$ contains the edge $f'_i \arcto f'_{i+1}$ for each index $a\le i<b$.  It follows that $C$ contains all faces $f'_i$, which implies that $\Sleeve \cap \ell$ is connected.  
\end{subproof}

It follows that $\Sleeve$ has two boundary curves, each of which is a simple $\theta$-monotone cycle in~$\Map$.  Let $\Sbdry$ denote the southern boundary of $\Sleeve$.  Any face $f$ in $\Sleeve$ that is incident to $\Sbdry$ has at least one out-neighbor in $\Dn\Map$ that is is not in $C$.  On the other hand, because $C$ is a strongly connected component of $\Dn\Map$, every face in $\Sleeve$ has at least one out-neighbor  in $\Dn\Map$ that is in $C$.  It follows that every face in $\Sleeve$ that is incident to the southern boundary of $\Sleeve$ has out-degree $2$ in $\Dn\Map$, and thus must be a down-face.  Moreover, every edge of $\Sbdry$ is a leg of a \emph{unique} down-face; in particular, every edge of $\Sbdry$ has a corresponding directed edge in $\DnVee\Map$.

Finally, consider the set $\DnVee{\Sbdry}$ of all such directed edges.  If any two edges in $\DnVee{\Sbdry}$ were directed toward each other, they would be the legs of a single down-face, which we just argued is impossible.  Thus, $\DnVee{\Sbdry}$ is a directed cycle in $\DnVee{\Map}$. This completes the proof of Lemma~\ref{L:shellable}.
\end{proof}

Each of the conditions (b), (c), (d), and (e) in Lemma~\ref{L:shellable} can be tested in $O(n)$ time using textbook graph algorithms.

Awartani and Henderson \cite{ah-sgts-87} explicitly consider spherical triangulations with  edges that lie along longitudes.  It is straightforward to extend Lemma~\ref{L:shellable} to such triangulations by \emph{including} both orientations of longitudinal edges in~$\Rt{\Map}$, \emph{excluding} both orientations of their dual edges from~$\Dn\Map$, and \emph{excluding} both orientations of vertical edges from~$\DnVee\Map$.  With this extension, Lemma~\ref{L:shellable} immediately implies a key result of Awartani and Henderson:

\begin{corollary}[Awartani and Henderson \cite{ah-sgts-87}]
\label{Co:AH-SeamCondition}
If $\Map$ has a longitudinal seam, then $\Map$ is longitudinally shellable.
\end{corollary}

\begin{proof}
Suppose $\Map$ has a seam along some longitude $\ell$.  Then the edges of $\Map$ that are contained in $\ell$ must comprise a single undirected path the north and south faces of $\Map$.  It follows that $\Rt\Map$ contains \emph{directed} paths along~$\ell$ from each polar face to the other.  We conclude that $\Map$ meets condition~(d) of Lemma \ref{L:shellable}.
\end{proof}

\subsection{The Awartani--Henderson Embedding}
\label{SS:AH-embedding}

\begin{theorem}[Awartani and Henderson \cite{ah-sgts-87}]
\label{Th:AH-embedding}
Every longitudinally shellable triangulation of the sphere is sinkable.  Moreover, given any longitudinally shellable triangulation $\Map$, we can compute a longitudinal morph from $\Map$ to a southern triangulation in $O(n)$ time.
\end{theorem}

\begin{proof}
Let $\Map$ be any longitudinally shellable triangulation.  Corollary \ref{C:weak-sink} implies that to prove~$\Map$ is sinkable, it suffices to construct a southern weak triangulation $\SouthMap$ that is longitudinally equivalent to $\Map$, in $O(n)$ time.  By projecting or scaling to the plane $z = -1$, we can think of $\SouthMap$ as a weak \emph{planar straight-line} triangulation that is $\theta$-equivalent to~$\Map$, meaning each vertex of $\SouthMap$ must lie on the ray from the origin determined by its longitude in $\Map$.  Because the edges of $\SouthMap$ are straight line segments, our algorithm only needs to specify the location $p'_i$ of each vertex $i$ in~$\SouthMap$.

Following Awartani and Henderson \cite[Theorem 3.3]{ah-sgts-87}, we construct $\SouthMap$ by embedding the \emph{faces} of $\Gamma$ one at a time, following an arbitrary longitudinal shelling order $f_1, f_2, \dots, f_{2n-4}$.  Throughout the construction, we maintain a weakly convex polygon $W$---that is, a polygon with convex interior whose vertices have interior angle \emph{at most} $\pi$---whose interior contains the origin. This polygon satisfies the following invariants for each vertex $i$, after faces $f_1, \dots, f_k$ have been embedded:

\begin{itemize}
\item
If $\AboveF{i}$ has been embedded, then $p'_i$ lies outside the interior of $W$.

\item
If in addition $\BelowF{i}$ has \emph{not} been embedded, then $p'_i$ is a vertex of $W$.
\end{itemize}

To start the construction, we embed the north face $f_1$ by fixing its vertices \emph{arbitrarily} on their respective rays, for example on the unit circle.  We also initialize $W$ to be the convex hull of these three points.  Then for each index $k$ from $2$ to $2n-5$, we proceed as follows:
\begin{itemize}
\item
If $f_k$ is an up-face, then all three vertices of $f_k$ have already been placed on the boundary of $W$, so we embed $f_k$ as the convex hull of its vertices.  Because $W$ is weakly convex, it contains the base segment of $f_k$.  After embedding $f_k$, the apex of~$f_k$ is no longer a vertex of $W$, and the base of $f_k$ becomes an edge of $W$.  (The vertices of $f_k$ may be collinear, in which case the interior of $W$ does not change.)

\item
On the other hand, if $f_k$ is a down-face, the base vertices of $f_k$ have already been placed on the boundary of $W$, but not the apex.  We place the apex of $f_k$ at the intersection of its ray and the base of $f_k$; thus, $f_k$ is embedded as a degenerate triangle.  The apex of $f_k$ becomes a new vertex of $W$, and the interior of $W$ does not change.
\end{itemize}
In both cases, it is easy to verify that the invariants on $W$ hold after $f_k$ is embedded.  Finally, when we consider the south face $f_{2n-4}$, all of its vertices have already been placed.
\end{proof}

Awartani and Henderson's construction can be described more concisely as follows.  We construct a southern weak triangulation $\SouthMap$ that is longitudinally equivalent to $\Map$, by considering vertices in any topological order of $\DnVee\Map$ and assigning a new coordinate $z'_i$ to each vertex.  (This topological order is an example of a \emph{canonical order} for $\Map$~\cite{fpp-hdpgg-90, cdf-cotca-18}.)  The first three vertices lie on the north face and thus can be placed anywhere below the equator on their respective longitudes.  For each later vertex $i\ge 4$, we make $z'_i$ as large as possible, such that no face induced by vertices~$1$ through $i$ (except the north face) is inverted.  Awartani and Henderson's argument implies that every down-face in the resulting triangulation $\SouthMap$ is degenerate.

% Even more concisely: Embed the vertices of the north face anywhere on their respective longitudes, and then embed all other vertices so that every down-face is degenerate.

\subsection{One-Bend Morphing}

Several authors have developed morphing algorithms for planar graphs that allow (or require) \emph{bends} in intermediate edges~\cite{lp-mpgdb-11,blps-mopgd-13,bls-mpgdo-24,ekp-ifmpg-03}.  To close this section, we describe an efficient morphing algorithm that introduces at most one bend into each edge; that is, we can efficiently construct an isotopy between any two isomorphic spherical triangulations, such that each edge of every intermediate triangulation is either a shortest path or the concatenation of two shortest paths.  We call such an isotopy a \emph{one-bend morph}.

Our algorithm uses a mild generalization of longitudinal morphing.  A \emph{rotated longitudinal morph} moves all vertices move along great circular arcs through some pair of antipodal points, called the \emph{poles} of the morph.  Rotated longitudinal morph are the natural spherical analogue of \emph{unidirectional morphs} in the plane, which move the vertices of a planar straight-line graph along parallel lines at constant speeds \cite{bhl-mpgdu-13,aabcd-hmpgd-17}.  Any rotated longitudinal morph can be implemented by linearly interpolating the homogeneous coordinate vectors of vertices along parallel lines in $\Real^3$, as described in the proof of Lemma~\ref{L:shell-sink}.

\begin{corollary}
Every pair of isomorphic $n$-vertex spherical triangulations is connected by a one-bend morph consisting of $O(1)$ rotations and $O(n)$ rotated longitudinal morphs.  Moreover, this one-bend morph can be computed in $O(n^2)$ time.
\end{corollary}

\begin{proof}
Let $\Map_0$ and $\Map_1$ be arbitrary isomorphic triangulations with the same underlying  planar graph $G$.  If necessary, we rotate both triangulations so that they have the same north face, and so that every great circle through the poles contains at most one vertex of each.  We then construct a one-bend morph from $\Map_0$ and $\Map_1$ through an intermediate \emph{southern} coherent triangulation $\Map_{1/2}$, similarly to our strategy in Section \ref{SS:strategy}.  Again, the intermediate triangulation $\Map_{1/2}$ can be constructed in $O(n)$ time \cite{dg-cop3d-97}.

To construct a one-bend morph from $\Map_0$ to $\Map_{1/2}$, we start by refining $\Map_0$ so that it contains a longitudinal seam, as shown in Figure \ref{F:seam}.  Let $\ell$ be the longitude through an arbitrary vertex of the north face of~$\Map_0$.  We refine $\Map_0$ by introducing bend vertices at the intersection of $\ell$ with each edge of $\Map_0$, and then adding edges that split the south face into two smaller triangles and each non-polar face intersecting $\ell$ into three smaller triangles.  Call the resulting refined triangulation $\tilde\Map_0$.  Awartani and Henderson's results (Corollary~\ref{Co:AH-SeamCondition} and Theorem~\ref{Th:AH-embedding}) imply that we can longitudinally morph~$\tilde\Map_0$ to a southern triangulation $\tilde\Map'_0$ in $O(n)$ time.

\begin{figure}[htb]
\centering
\includegraphics[scale=0.2]{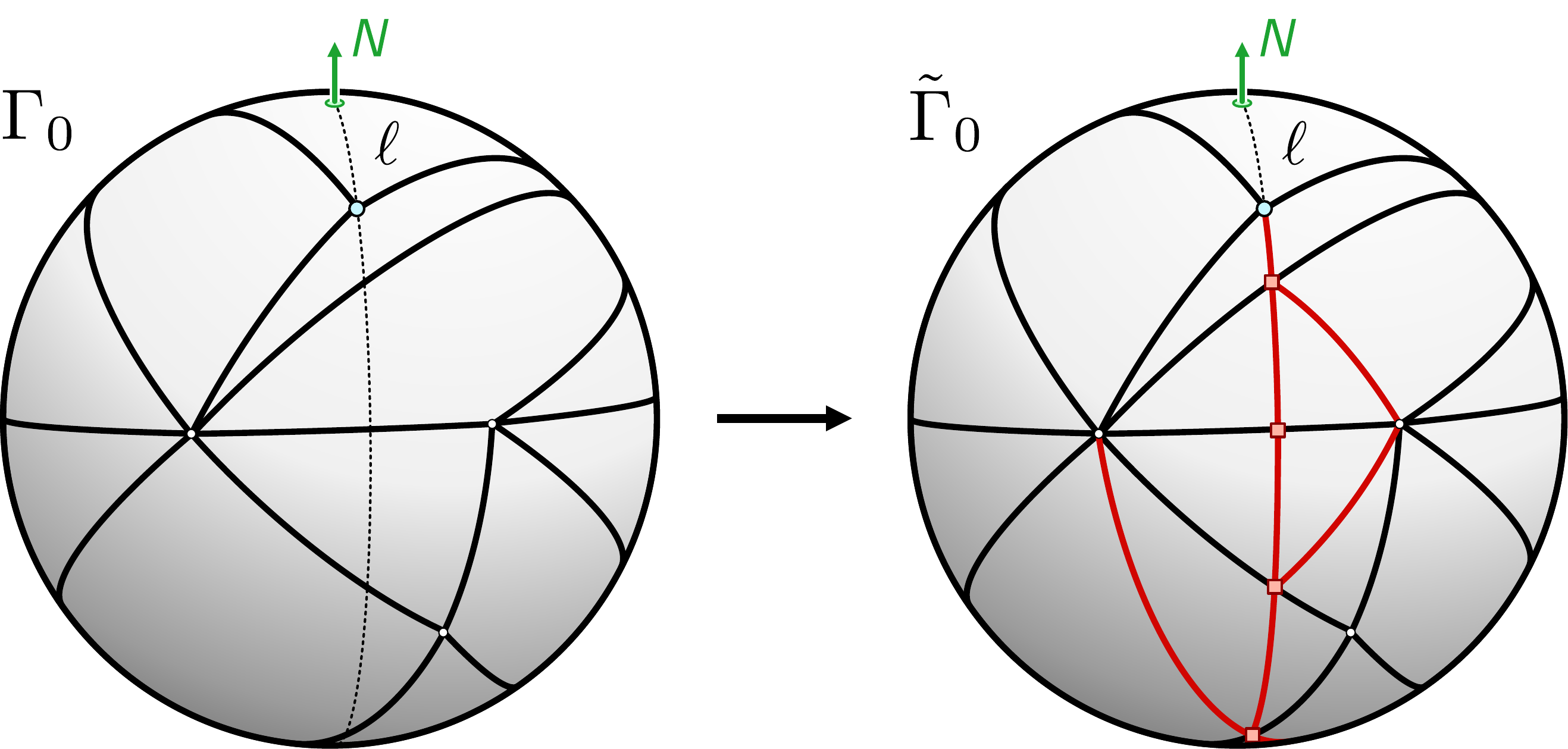}
\caption{Refining a triangulation along a longitude.}
\label{F:seam}
\end{figure}

We then refine $\Map_{1/2}$ into a southern triangulation $\tilde\Map_{1/2}$ isomorphic to $\tilde\Map'_0$, by bisecting the same edges and subdividing the same faces as we did to refine $\Map_0$.  Finally, we can construct a morph between the southern triangulations $\tilde\Map'_0$ and $\tilde\Map_{1/2}$ using any planar morphing algorithm.  In particular, an algorithm of Alamdari \etal~\cite{aabcd-hmpgd-17} computes a planar morph consisting of a sequence of $O(n)$ unidirectional morphs; an improvement by Klemz \cite{k-cdhgl-21} computes this morph in $O(n^2)$ time.  Lifting a unidirectional morph in the plane $z=-1$ to the sphere yields a rotated longitudinal morph whose poles lie on the equatorial plane $z=0$.

Putting the pieces together, we obtain a morph from the refined triangulation $\tilde\Map_0$ to the refined coherent triangulation $\tilde\Map_{1/2}$ that consists of $O(n)$ rotated longitudinal morphs.  Ignoring the refinement edges yields a morph from $\Map_0$ to $\Map_{1/2}$ where each intermediate edge has at most one bend.  We can extend this morph to $\Map_1$ by concatenating the reversal of a similar morph from $\Map_1$ to $\Map_{1/2}$.  The overall running time is dominated by two invocations of Klemz's planar morphing algorithm \cite{k-cdhgl-21}.
\end{proof}

Theorem \ref{Th:3con} immediately generalizes this result to 3-connected embeddings:

\begin{corollary}
Every pair of isomorphic shortest-path embeddings of the same 3-connected planar graph on the sphere is connected by a one-bend morph consisting of $O(1)$ rotations and $O(n)$ rotated longitudinal morphs.  Moreover, this one-bend morph can be computed in $O(n^2)$ time.
\end{corollary}

% ------------------------------------------------------------------------------
\section{Longitudinally Shellable Rotations}
\label{S:shellrot}

When a given triangulation $\Map$ is not longitudinally shellable, we can still apply our morphing strategy if we can find a rotation of the sphere that transforms $\Map$ into a longitudinally shellable triangulation.  In this section, we describe an efficient algorithm that either finds such a rotation or correctly reports that no such rotation exists.  Experimental evidence (described in Section~\ref{S:experiment}) suggests that “in practice”, even for triangulations with many long skinny triangles, a significant fraction of rotations of $\Gamma$ are longitudinally shellable, so it suffices to consider a small number of random rotations.  In fact, without exception, every one of the thousands of random adversarial triangulations we generated had shellable rotations.  Nevertheless, in Section~\ref{SS:no-rot}, we construct a triangulation with no longitudinally shellable rotation. 

\subsection{Finding a Shelling Direction}

Instead of considering different rotations of $\Map$, it is convenient to keep $\Map$ fixed and consider different locations for the “north pole”.  We call a unit vector $p$ a \EMPH{shelling direction} for~$\Map$ if any (and therefore every) rotation of the sphere that takes $p$ to the standard north pole $(0,0,1)$ also takes $\Map$ to a longitudinally shellable triangulation.  We similarly define the directed graphs $\Dn{\Map}(p)$ and $\Rt{\Map}(p)$ and $\DnVee\Map(p)$.  For example, for any unit vector $p = (x,y,z)$, the graph $\Rt{\Map}(p)$ contains the edge $\arc{i}{j}$ if and only if 
\[
	\vol(p,i,j) \coloneq
	\det\begin{pmatrix}
		x & y & z \\
		x_i & y_i & z_i \\
		x_j & y_j & z_j \\
	\end{pmatrix}
	> 0,
\]
and $p$ is a shelling direction for $\Map$ if and only if the graphs $\Dn{\Map}(p)$ and $\DnVee{\Map}(p)$ are acyclic.

\begin{theorem}
\label{Th:shellrot}
Given any shortest-path triangulation $\Map$ on the sphere, we can either compute a shelling direction for $\Map$ or report correctly that no such direction exists, in $O(n^{2.5}\log^3 n)$ time.
\end{theorem}

\def\Cell{C}

\begin{proof}
Each edge $ij$ in $\Map$ lies on a unique great circle; let $\mathcal{A}(\Map)$ denote the arrangement of all $3n-6$ such great circles.  For any two unit vectors $p$ and $q$, the graphs $\Rt\Map(p)$ and~$\Rt\Map(q)$ are identical if and only if $p$ and~$q$ lie in the same cell of $\mathcal{A}(\Gamma)$.  Thus, for any cell~$\Cell$, we can write $\Rt\Map(\Cell)$ to denote the graph $\Rt\Map(p)$ for any $p\in \Cell$.  If $\Cell$ and $\Cell'$ are adjacent two-dimensional cells of $\mathcal{A}(\Gamma)$, then $\Rt\Map(\Cell)$ and $\Rt\Map(\Cell')$ differ in the direction of exactly one edge. 

We can find a shelling direction for $\Map$ in $O(n^3)$ time by constructing the great-circle arrangement $\mathcal{A}(\Map)$, and then for each two-dimensional cell $C$, check whether $\Rt\Map(C)$ is strongly connected. %
(This arrangement is centrally symmetric, and the gnomonic projection of either hemisphere to any tangent plane is an arrangement of \emph{lines}.  Thus, we can use any classical algorithm to construct line arrangements \cite{cgl-pgd-85,eg-tsa-89,eos-calha-86} instead of more sophisticated algorithms to construct arrangements of more general circles on spheres \cite{cl-cacsa-09}.) 

We can speed up this naive algorithm using a data structure of Diks and Sankowski for reachability queries in dynamic directed plane graphs~\cite{ds-tptc-07}.  Diks and Sankowski's data structure maintains a directed planar graph with a fixed embedding, supports edge insertions and deletions that do not change the embedding in $O(\sqrt{n}\log^3 n)$ time, and supports queries of the form “Is there a directed path from vertex $i$ to vertex $j$?” in $O(\sqrt{n}\log^2 n)$ time.

To use this data structure, we traverse the dual graph of the arrangement $\mathcal{A}(\Map)$.  At each two-dimensional cell $\Cell$, we perform two reachability queries in $\Rt\Map(\Cell)$ between the two polar faces.  If both reachability queries succeed, then by Lemma \ref{L:shellable}(d), cell $\Cell$ contains a shelling direction.  When we move from a cell $\Cell$ to a neighboring cell $\Cell'$, we delete one edge of $\Rt\Map(C)$ and insert its reversal.  Altogether, we perform $O(n^2)$ queries and $O(n^2)$ updates.
\end{proof}

\subsection{Unshellable From Every Direction}
\label{SS:no-rot}

Regard each edge of the undirected dual graph $\Map^*$ as a pair of opposing directed edges, which we call \EMPH{darts}.  For any directed cycle $\gamma^*$ of darts in $\Map^*$, we define the \EMPH{polar region $P(\gamma^*)$} to be the set of all north poles $p$ such that the directed dual graph $\Dn{\Map}(p)$ contains every dart in $\gamma^*$.  A unit vector $p$ is a shelling direction for $\Map$ if and only if $p \not\in P(\gamma^*)$ for \emph{every} directed dual cycle $\gamma^*$.

\begin{lemma}
For every directed cycle $\gamma^*$ of darts in the dual graph $\Map^*$, the corresponding polar region $P(\gamma^*)$ is the interior of a convex spherical polygon, that is, the intersection of~$S^2$ with a (possibly empty) open convex polyhedral cone.
\end{lemma}

\begin{proof}
Let $\arc{i}{j}$ be any dart in $\Map$, and let $f$ and $g$ be the faces incident to $\arc{i}{j}$ on the left and right, respectively.  The directed graph $\Dn\Map(p)$ contains the edge $f\arcto g$ if and only if $\vol(p,i,j) > 0$.  The set of all points $p$ satisfying this inequality is an open hemisphere $H(\arc{f}{g})$ bounded by the great circle through $\arc{i}{j}$.  Finally, the polar region $P(\gamma^*)$ of any dual cycle $\gamma^*$ is the intersection of the hemispheres $H(\arc{f}{g})$ for all $\arc{f}{g}\in\gamma^*$.
\end{proof}

A \EMPH{rotor} is the subset of faces of $\Map$ dual to any directed cycle $\gamma^*$ in the dual graph~$\Map^*$.  We sketch the construction of a triangulation $\Map$ containing a small number of rotors, whose polar regions cover the entire sphere.  We describe our construction in terms of an arbitrary sufficiently small angular parameter $\e>0$; in practice, it suffices to set $\e \approx 0.01 \approx 0.5^\circ$.  

First we define a simple \emph{equatorial rotor}, which triangulates a belt of width $O(\e^2)$ around a great circle using $O(1/\e)$ isosceles triangles with aspect ratio $O(\e)$ and with all edges at angle $O(\e)$ from the great circle.  The triangulation edges are consistently oriented so that for any north pole $p$ sufficiently far from the rotor, the rotor defines a cycle in the directed dual graph $\Dn\Map(p)$.  The polar region of (one orientation of) the rotor is a convex spherical polygon whose boundary has Hausdorff distance $O(\e)$ from the rotor itself.  In particular, in the limit as $\e$ approaches zero, this polar region approaches an open  hemisphere.

\begin{figure}[htb]
\centering
\includegraphics[scale=0.4]{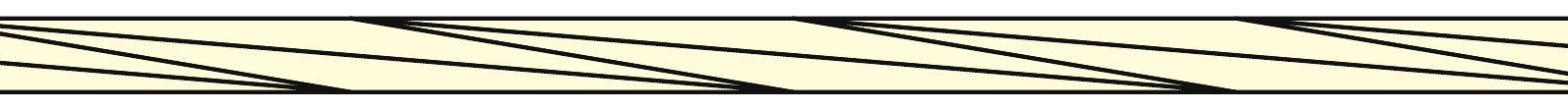}
\caption{Local structure of an equatorial rotor.}
\end{figure}

Our unshellable triangulation contains four modified equatorial rotors, whose central great circles are parallel to the faces of a regular tetrahedron.  Let~$Q$ denote the arrangement of these four great circles; $Q$ is the central projection of an inscribed semi-regular cuboctahedron; see Figure \ref{F:qbok}.  Each pair of great circles meets at an angle of $\arccos(1/3) \approx 70.529^\circ$.

\begin{figure}[htb]
\centering
\includegraphics[scale=0.4]{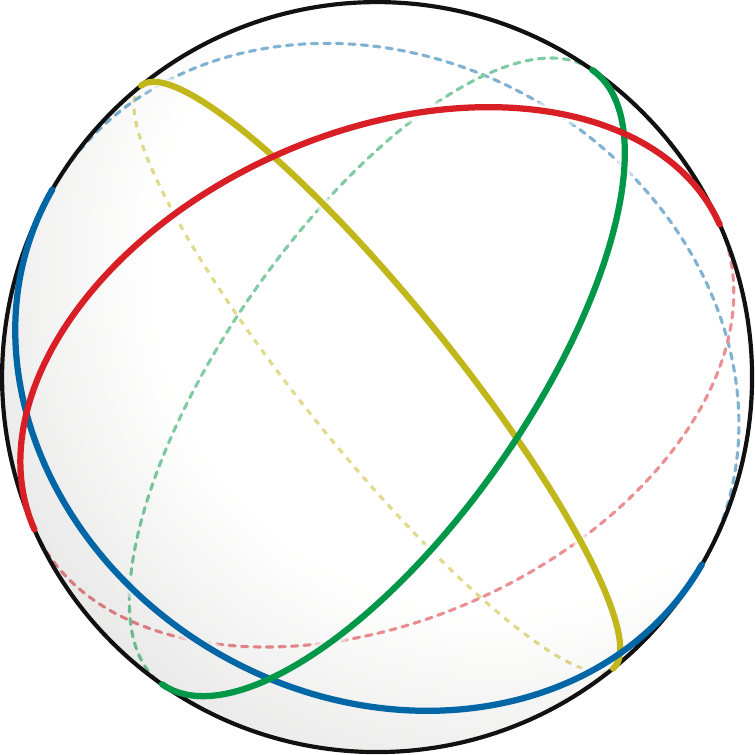}
\caption{Four great circles defining a spherical cuboctahedron.}
\label{F:qbok}
\end{figure}

%Because every vertex in the dual graph $T^*$ has degree $3$, if any two cycles in $T^*$ intersect, they must share at least one edge.  (Moreover, if two cycles in $T^*$ actually cross, they must share at least two edges.)  Thus, it is impossible for two equatorial rotors to cross unless the angle between their defining great circles is at most $O(\e)$.

At each vertex of $Q$, we modify the two equatorial rotors by introducing a \emph{crossing gadget}, illustrated in Figure \ref{F:xgadget}(a).  Overall the gadget resembles a rhombus with diameter $O(\e)$ whose edges are parallel to the great circles defining the rotors.  Each of the equatorial rotors is broken and offset by $O(\e)$ to cover two opposite edges of this rhombus.  Then three new edges are added to reconnect both rotors; these edges have distance and angle $O(\e)$ from the long diagonal of the rhombus and the bisector of the smaller angle between the two great circles.  The two fat triangles inside the rhombus are incorporated into one of the two rotors (the vertical red rotor in Figure \ref{F:xgadget}); the two thin triangles near the diagonal are incorporated into both.  Figure \ref{F:no-shell-plan}(b) shows a schematic of our overall construction.

\begin{figure}[htb]
\centering\scriptsize\sffamily
\begin{tabular}{c@{\quad}c}
\raisebox{-0.5\height}{\includegraphics[scale=0.29]{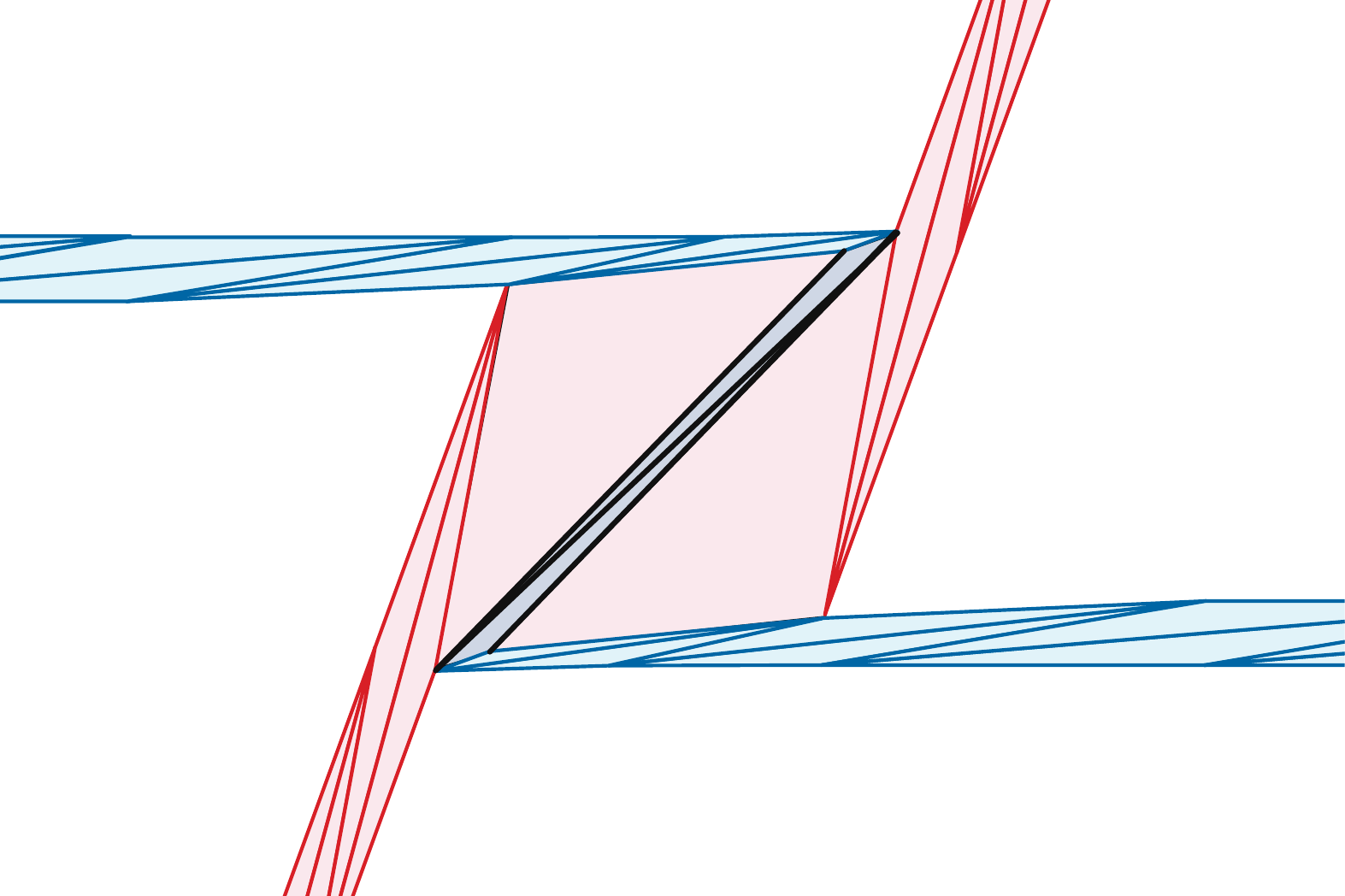}} &
\raisebox{-0.5\height}{\includegraphics[scale=0.4]{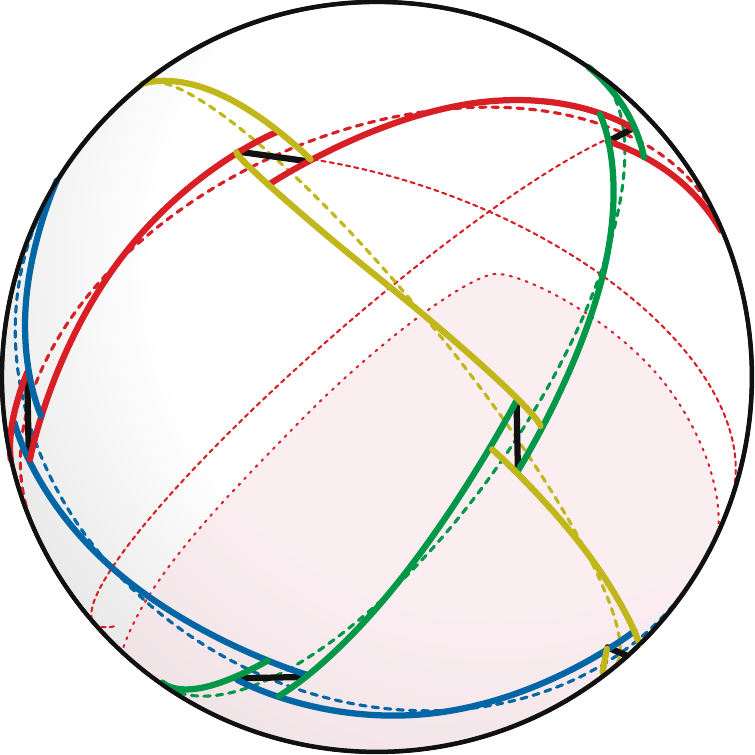}} \\
(a) & (b)
\end{tabular}
\caption{(a)~A single crossing gadget.  (b)~Constructing a triangulation with no shelling direction.  Each color of broken lines is an equatorial rotor; each short black line is the diagonal of a crossing gadget. The polar region of one (red) equatorial rotor is shaded.}
\label{F:xgadget}
\label{F:no-shell-plan}
\end{figure}

Now let $\gamma^*$ be a directed cycle in the dual graph $\Map^*$ through the faces of one modified equatorial rotor.  The corresponding polar region $P(\gamma^*)$ nearly fills a triangular region bounded by three bisector circles, each defined by an antipodal pair of crossing gadgets.  This polar region completely covers one triangular face and nearly half of three square faces of $Q$.  The reversal of~$\gamma^*$ defines an antipodally symmetric polar region.

Thus, the union of the polar regions defined by all four equatorial rotors covers the entire sphere except for a small “hole” near the center of each square face of $Q$, which we can make arbitrarily small by choosing $\e$ appropriately.  We cover these holes by adding a small rotor, reminiscent of the projection of a Schönhardt polyhedron, inside each square face of~$Q$; see Figure \ref{F:square-rotor}.  Finally, to complete the triangulation $\Map$, we arbitrarily triangulate the area between the rotors.  

\begin{figure}[htb]
\centering
\includegraphics[scale=0.4]{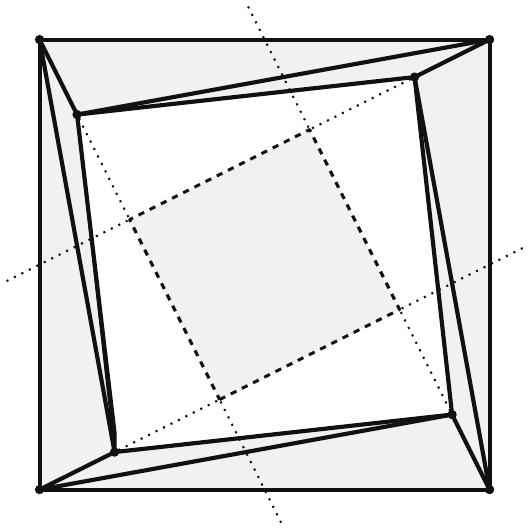}
\caption{A square rotor (solid lines) and its polar region (dashed lines)}
\label{F:square-rotor}
\end{figure}

\begin{theorem}
There is a shortest-path triangulation of the sphere with no longitudinal shelling direction.
\end{theorem}

\begin{proof}
By construction, every point $p\in S^2$ lies in the polar region of at least one rotor, and therefore in the polar region of at least one directed cycle in the dual graph $\Map^*$.  Thus, for every $p\in S^2$, the directed dual graph $\Dn\Map(p)$ is not acyclic.  The theorem follows immediately from Lemma \ref{L:shellable}.
\end{proof}

Nothing in our construction prevents our bad triangulation $\Map$ from being \emph{sinkable}.  In particular, just before we add the final rotors to cover the “holes” in the square faces of $Q$, the partial triangulation is still longitudinally shellable from any direction inside one of those holes.  After rotating and sinking this partial triangulation and then projecting to the plane $z=-1$, we can add the final rotors to the resulting planar embedding.  Moreover, experimental evidence  strongly suggests that every triangulation in which all edges have length less than $\pi/2$ is sinkable; see Conjecture~\ref{C:acute}.  Our bad triangulation $\Map$ satisfies this condition.

% ------------------------------------------------------------------------------
\section{Sinking}
\label{S:sink}

Finally, we present an exact characterization of sinkable triangulations.  Let $\InitMap$ be any shortest-path triangulation of the sphere.  Recall from Corollary~\ref{C:weak-sink} that $\InitMap$ is sinkable if and only if there is a weak triangulation $\SouthMap$, with all vertices below the equator, that is $\theta$-equivalent to $\InitMap$.

By scaling coordinate vectors if necessary, as described in Section \ref{SS:coords},  we can consider only weak triangulations $\SouthMap$ where each vertex has the same $x$- and $y$-coordinates as in $\Map$.  Let $(x_i, y_i, \InitZ_i)$ and $(x_i, y_i, \SouthZ_i)$ denote the coordinates of vertex $i$ in $\InitMap$ and $\SouthMap$, respectively.  Then $\SouthMap$ is $\theta$-equivalent to $\InitMap$ if and only if $\South{\vol}(i,j,k) \ge 0$ for every non-polar face $(i,j,k)$, where
\[
	\South{\vol}(i,j,k) \coloneq \det\begin{pmatrix}
		x_i & y_i & \SouthZ_i \\
		x_j & y_j & \SouthZ_j \\
		x_k & y_k & \SouthZ_k
	\end{pmatrix}
\]
We reiterate that because all $x$- and $y$-coordinates are fixed, the volume determinant $\vol'(i,j,k)$ is a linear function of the vector $\SouthZ$.

\begin{lemma}
\label{L:sink-LP}
A spherical triangulation $\InitMap$ is sinkable if and only if the following linear program is feasible:
\begin{equation}
	\def\arraystretch{1.2}
	\begin{array}{r@{\quad}r@{~}l@{\qquad}l}
	\text{maximize}
		& \sum_i \SouthZ_i \\
	\text{subject to}
		& \SouthZ_i & = -1 
			& \text{for every vertex $i$ of $\NorthF$} \\
		& \South{\vol}(i,j,k) & \ge 0
			& \text{for every non-polar face $(i,j,k)$}
	\end{array}
	\label{Eq:LP}
\end{equation}
\end{lemma}

\begin{proof}
Suppose $\SouthZ = (\SouthZ_1, \dots, \SouthZ_n)$ is the $z$-coordinate vector of a weak triangulation $\SouthMap$ that is $\theta$-equivalent to $\InitMap$, such that $\SouthZ_i < 0$ for all $i$.  Without loss of generality, assume the vertices of the north face of $\SouthMap$ are vertices $1$, $2$, and~$3$.  We can move these three vertices to the plane $z = -1$ by applying the linear transformation $(x,y,z) \mapsto (x, y, ax+by+cz)$, where $a$, $b$, and $c$ satisfy the linear system
\[
	\begin{pmatrix}
		x_1 & y_1 & z'_1 \\
		x_2 & y_2 & z'_2 \\
		x_3 & y_3 & z'_3 \\
	\end{pmatrix}
	\begin{pmatrix} a \\ b \\ c \end{pmatrix}
	=
	\begin{pmatrix} -1 \\ -1 \\ -1 \end{pmatrix}.
\]
This linear transformation has positive determinant.  Let $\NewSouthZ$ denote the transformed $z$-coordinate vector and $\NewSouthMap$ the corresponding weak triangulation.  Because no face of $\NewSouthMap$ is inverted, the coordinate vector of every other vertex $i$ is a positive linear combination of the north-face coordinate vectors; thus, we have $\NewSouthZ_i < 0$ for all~$i$.
\end{proof}

Lemma \ref{L:sink-LP} immediately implies an algorithm to compute a sinking longitudinal morph, or report correctly that none exists, in weakly polynomial time, provided the input $x$- and $y$-coordinates are rational \cite{gls-gaco-93}.  We obtain a simpler and faster algorithm by identifying the optimal basis of linear program \eqref{Eq:LP}, subject to a mild technical condition.

First we need a small extension of a result used in many planar morphing algorithms \cite{hk-prga-92,epln-sdahg-06,k-cdhgl-21,cgt-cdg23-96,kklss-cimpg-19}.  Recall that a polygon is \emph{weakly convex} if its interior is convex and each of its vertices has interior angle \emph{at most} $\pi$.  Let $T$ be any straight-line plane graph whose outer face is a simple $x$-monotone polygon and whose inner faces are triangles.  A weakly convex polygon~$W$ is \emph{compatible} with~$T$ if there is a homeomorphism from the boundary of~$T$ to the boundary of~$W$ that preserves $x$-coordinates.  A weak planar triangulation~$T'$ is \emph{$x$-equivalent} to~$T$ if it has the same underlying graph, every nondegenerate face of $T'$ has the same orientation as the corresponding face of~$T$, and corresponding vertices have equal $x$-coordinates. 

\begin{lemma}
\label{L:weak-hn}
Given a triangulation $T$ of an $x$-monotone polygon in the plane, and a weakly convex polygon~$W$ that is compatible with $T$, there is a weak planar triangulation $T'$ that is $x$-equivalent to~$T$ and whose outer face is $W$.
\end{lemma}

\begin{proof}
When $W$ is a \emph{strictly} convex polygon (all interior angles are strictly less than $\pi$), the lemma follows immediately from  algorithms of Chrobak, Goodrich, and Tammassia \cite[Theorem 3.5]{cgt-cdg23-96} and Kleist, Klemz, Lubiw, Schlipf, Staals, and Strash~\cite[Lemma 11]{kklss-cimpg-19}.  Both of these algorithms first compute positive weights $\lambda_{\arc{i}{j}} > 0$ for each dart $\arc{i}{j}$ of $T$ so that the $x$-coordinate of each interior vertex $j$ is the weighted average of the $x$-coordinates of the neighbors of $j$.  For each interior vertex $j$, the incoming dart weights satisfy the constraints
\[
	x_j = \sum_i \lambda_{\arc{i}{j}} x_i
	\quad\text{and}\quad
	\sum_i \lambda_{\arc{i}{j}} = 1,
\]
where for notational simplicity we define $\lambda_{\arc{i}{j}} = 0$ if $ij$ is not an edge of $T$.  Chrobak \etal~compute the weights~$\lambda_{\arc{i}{j}}$ using flows \cite[Lemma 3.4]{cgt-cdg23-96}; Kleist \etal~use a simpler averaging computation at each vertex $i$~\cite[Lemma 11]{kklss-cimpg-19}.

Then, to compute the new embedding $T'$, both algorithms fix the boundary vertices to the corresponding vertices of $W$, and then compute the interior $y$-coordinates by solving the linear system
\begin{equation}
	y'_j = \sum_i \lambda_{\arc{i}{j}} y'_j \qquad\text{for every interior vertex $j$}
	\label{Eq:Tutte}
\end{equation}
Floater's extension \cite{f-ptsda-98} of Tutte's spring embedding theorem \cite{t-hdg-63} guarantees that the resulting $y$-coordinates define a proper straight-line triangulation $T'$.

A straightforward limiting argument implies that when $W$ is a \emph{weakly} convex polygon, the drawing $T'$ produced by either of these algorithms is a \emph{weak} triangulation.  Consider a continuous family $W(t)_{t\ge 0}$ of polygons, where $W(0)=W$ and $W(t)$ is strictly convex for all $t>0$.  For each $t\ge 0$, let $T'(t)$ denote the drawing computed by solving linear system~\eqref{Eq:Tutte} with all boundary vertices fixed to $W(t)$.  This linear system is always non-singular, so $T'(t)$ is well-defined and varies continuously with $t$.  In particular, the signed area of each face of $T'(t)$ varies continuously with $t$ and is positive for all $t>0$.  Thus, $T'(0)$ may contain degenerate faces, but no inverted faces.  We conclude that $T' = T'(0)$ is a weak triangulation $x$-equivalent to~$T$.
\end{proof}

We emphasize that our algorithms never \emph{compute} the weak triangulation $T'$; we only need to prove that it exists.

\begin{theorem}
\label{Th:sinksys}
If $\SouthZ = (\SouthZ_1, \dots, \SouthZ_n)$ is an optimal solution to linear program \eqref{Eq:LP}, then $z'$ is also a solution to the following $n\times n$ linear system:
\begin{equation}
	\def\arraystretch{1.2}
	\begin{array}{r@{~}l@{\qquad}l}
		z'_i & = -1 
			& \text{for every vertex $i$ of $\NorthF$} \\
		\vol'(i,j,k) & = 0
			& \text{for every down-face $(i,j,k)$}
	\end{array}
	\label{Eq:Eq}
\end{equation}
\end{theorem}

\def\YProject{\pi_y}
\def\Trapezoid{\tau}

\begin{proof}
Let $\SouthZ$ be any feasible solution for our linear program.  As in the proof of Lemma \ref{L:sink-LP}, the coordinate vector of every vertex $i$ is a positive linear combination of the north-face coordinate vectors, so $\SouthZ_i < 0$ for all $i$.  Let $\SouthMap$ be the southern weak triangulation defined by $\SouthZ$, and suppose that at least one down-face of $\SouthMap$ is nondegenerate.  We argue that we can construct another southern weak triangulation~$\NewSouthMap$ (that is, another feasible solution $\NewSouthZ$ to \eqref{Eq:LP}) by increasing some $z$-coordinates and leaving all other coordinates fixed, which implies that $\SouthZ$ is not an optimal solution to \eqref{Eq:LP}.

Call a vertex of $\SouthMap$ \emph{sober} if it is not incident to the north face or to any degenerate face.  If any vertex of $\SouthMap$ is sober, we can construct $\NewSouthMap$ by moving any sober vertex upward slightly and keeping all other vertices fixed.  So without loss of generality, we assume that $\SouthMap$ has no sober vertices.
We call a vertex $i$ \emph{upward-free} if $i$ is not a vertex of the north face and $\AboveF{i}$ is nondegenerate in $\SouthMap$, \emph{downward-free} if $\BelowF{i}$ is nondegenerate in $\SouthMap$, \emph{totally free} if it is both upward- and downward-free, and \emph{trapped} if $\AboveF{i}$ and $\BelowF{i}$ are both degenerate.

A \emph{bar} in $\SouthMap$ is a maximal circular arc that is covered by edges of $\SouthMap$ and has no totally free vertices in its interior.  Each bar is either a single edge or the union of one or more degenerate faces; in the latter case, the corresponding subset of faces in $\Map$ is edge-connected, and its union is a $\theta$-monotone disk.  Equivalently, a nontrivial bar is the image in $\SouthMap$ of a maximal edge-connected subset of faces of $\Map$ that are all degenerate in $\SouthMap$.  We call the subcomplex of~$\Map$ induced by these faces a \emph{pre-bar}.  The \emph{endpoints} of a pre-bar are the preimages of the endpoints of its bar.  Because no vertex of $\SouthMap$ is sober, every vertex not incident to the north face is contained in at least one bar.

Every upward-free vertex in $\SouthMap$ is the image of a vertex on the northern boundary of a pre-bar in $\Map$.  Every downward-free vertex in $\SouthMap$ is either incident to the north face or the image of a vertex on the southern boundary of a pre-bar in $\Map$.  Each endpoint of a bar can be upward-free, downward-free, or both.  Finally, each trapped vertex in $\SouthMap$ is the image of a vertex in the interior of a pre-bar in $\Map$, and thus lies on a unique bar in $\SouthMap$.  See Figure \ref{F:perturb-bar}(a) and (b).

\begin{figure}[htb]
\centering\footnotesize\sffamily
\includegraphics[scale=0.4]{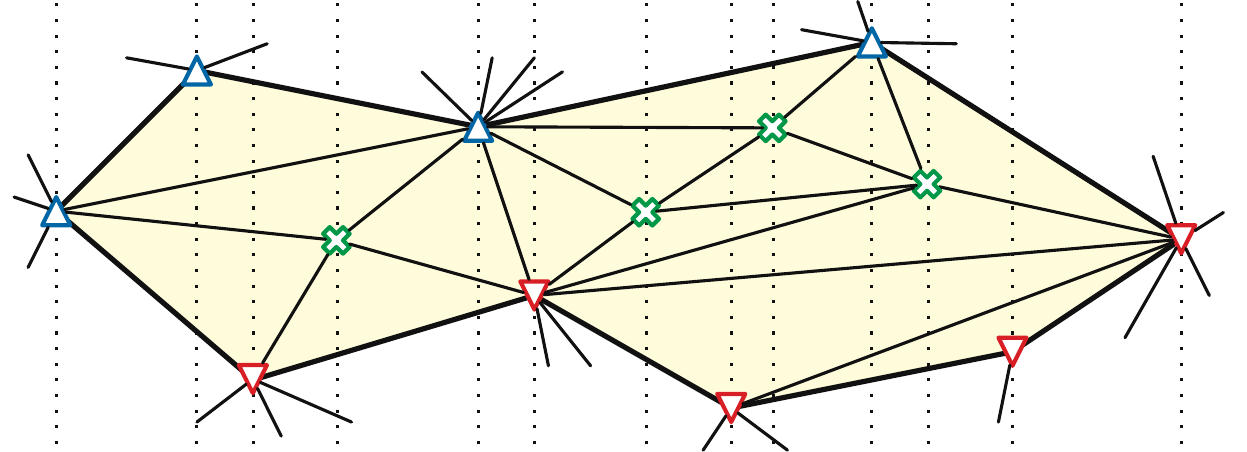}\\[-1ex](a)\\[3ex]
\includegraphics[page=2,scale=0.4]{Fig/perturb-bar}\\(b)\\[2ex]
\includegraphics[page=3,scale=0.4]{Fig/perturb-bar}\\[-1ex](c)
\caption{Planar projections of (a) a pre-bar in $\Map$, (b) the corresponding bar in $\SouthMap$, and (c) the corresponding perturbed bar in $\NewSouthMap$ (with curved edges to show degenerate facets).  Triangles indicate upward- and downward-free vertices; \textsf{X}'s indicate trapped vertices.  Dotted lines indicate longitudes.}
\label{F:perturb-bar}
\end{figure}

We construct a new southern weak triangulation $\NewSouthMap$ by defining a new coordinate vector $\NewSouthZ$, such that $\SouthZ_i \le \NewSouthZ_i < 0$ for each vertex~$i$.  Let $\e>0$ be any sufficiently small positive constant.  Defining new coordinates for the free vertices of $\SouthMap$ is straightforward:
\begin{itemize}
\item
For every vertex $i$ of the north face, we define $\NewSouthZ_i = \SouthZ_i = -1$.
\item
For every upward-free vertex $i$, we define $\NewSouthZ_i = \SouthZ_i/(1+\e) > \SouthZ_i$.
\item
For every downward-free vertex $i$ that is not also upward free, we define $\NewSouthZ_i = \SouthZ_i$.
\end{itemize}
Choosing sufficiently small $\e$ already guarantees that any face that is non-degenerate in~$\SouthMap$ is also non-degenerate in $\NewSouthMap$.  However, placing the trapped vertices to avoid inverting degenerate faces of $\SouthMap$ requires more care.

Let $\beta$ be any nontrivial bar in $\SouthMap$, and let~$\PreBar$ be the corresponding pre-bar in $\Map$.  We have already mapped the entire boundary of~$\PreBar$ to a spherical trapezoid~$\Trapezoid$ in $\NewSouthMap$.  Specifically, if $\Bar$ lies on the plane $z = ax+by$, then all upward-free vertices in $\Bar$ lie on the plane $z = (ax+by)/(1+\e)$ in~$\NewSouthMap$.  

Because $\Bar$ projects to a line segment in the plane $z = -1$, it lies inside an open hemisphere bounded by a \emph{vertical} plane through the origin.  Without loss of generality, suppose $\Bar$ lies in the hemisphere $y > 0$.  Let $\YProject$ denote the gnomonic projection $(x,y,z) \mapsto (x/y, z/y)$ from $S^2$ onto the plane $y=1$.  Then $\YProject(\Bar)$ is a non-vertical line segment, $\YProject(B)$ is a triangulation of an $x$-monotone polygon, and $\YProject(\Trapezoid)$ is a planar trapezoid whose lower boundary is a subset of $\YProject(\beta)$ and whose upper boundary lies on a line parallel to $\YProject(\beta)$.  Each longitude $\ell$ through a trapped vertex $i$ on $\beta$ defines a \emph{vertical} line $\YProject(\ell)$ that must contain the point $(x_i/y_i, \NewSouthZ_i/y_i)$.  See Figure \ref{F:perturb-bar}.

To summarize, $\YProject(\PreBar)$ is a triangulation of an $x$-monotone polygon, and $\YProject(\Trapezoid)$ is a weakly convex polygon (specifically, a trapezoid) that is compatible with $\YProject(\PreBar)$.  Thus, Lemma \ref{L:weak-hn} implies that there is a weak triangulation $\YProject(\NewSouth{T})$ that is $x$-equivalent to $\YProject(\PreBar)$ and whose outer face is $\YProject(\Trapezoid)$.  Pulling $\YProject(\NewSouth{T})$ back to the sphere gives us a weak triangulation $\NewSouth{T}$ of the spherical trapezoid $\Trapezoid$ that is $\theta$-equivalent to $B$.  Each interior vertex $i$ of $\PreBar$ maps to a vertex in the closed interior of~$\Trapezoid$, which implies $\SouthZ_i \le \NewSouthZ_i \le \SouthZ_i/(1+\e)$.

We assemble the overall weak triangulation $\NewSouthMap$ by applying this construction to every bar in $\SouthMap$.  (Processing each bar requires projecting to a different vertical tangent plane.)  Because $\NewSouthZ_i \ge \SouthZ_i$ for every vertex $i$, and $\NewSouthZ_i > \SouthZ_i$ for at least one vertex $i$, the original weak triangulation~$\SouthMap$ cannot be an optimal solution to our linear program \eqref{Eq:LP}.
\end{proof}

\begin{corollary}
Given a shortest-path triangulation $\Map$ of the sphere, we can either compute a longitudinal morph from $\Map$ to a southern triangulation, or report correctly that no such morph exists, in $O(n^{\omega/2})$ time, assuming linear system \eqref{Eq:Eq} is non-singular.
\end{corollary}

\begin{proof}
The support graph of linear system \eqref{Eq:Eq} is precisely the directed planar graph $\DnVee{\Map}$.  Thus, assuming the system is non-singular, it can be solved in $O(n^{\omega/2}) = O(n^{1.1864})$ time (in the real RAM model) via nested dissection and fast matrix multiplication \cite{lrt-gnd-79,ay-msnda-13}.  If the solution~$z'$ is a feasible point for linear program \eqref{Eq:LP}, the corresponding drawing $\SouthMap$ is a weak southern triangulation longitudinally equivalent to $\Map$.  If the system is non-singular, and the solution~$z'$ is not a feasible point for linear program \eqref{Eq:LP}, then by Theorem \ref{Th:sinksys}, the linear program is infeasible, so we can correctly report that $\Map$ is not sinkable.
\end{proof}

Theorem \ref{Th:sinksys} is a strict generalization of the Awartani--Henderson embedding described in Section~\ref{SS:AH-embedding}.  Recall from Lemma \ref{L:shellable} that $\Map$ is longitudinally shellable if and only if  its support graph $\DnVee\Map$ is acyclic.  Thus, we can permute the rows and columns of \eqref{Eq:Eq} to obtain an upper-triangular system, which we can then solve by back-substitution.  Each step of back-substitution assigns a $z$-coordinate to the apex of a single down-face, exactly mirroring one step of Awartani and Henderson's construction.

% ----------------------------------------------------------------------
\section{Experimental Results}
\label{S:experiment}

We implemented a suite of algorithms to construct spherical triangulations, test shellability and sinkability, construct Awartani--Henderson embeddings, and visualize longitudinal morphs.  (This implementation was invaluable in identifying the objective function for our linear program \eqref{Eq:LP} and leading us to conjecture Theorem \ref{Th:sinksys}.)  In particular, to stress-test our algorithms, we implemented several randomized heuristics to construct “ugly” triangulations with long edges and skinny triangles, including generalizations of Schönhardt's polyhedron \cite{s-uzdt-28,s-pdp-48,s-pdp2-51} and Jessen's icosahedron \cite{j-oi-67,d-sb-71}, convex hulls of random points, equatorial rotors, edge flipping and other refinement operations, and more severely, inserting arbitrary shortest paths as new edges.

We tested several thousand random triangulations, with dozens to hundreds of vertices.  Without exception, every triangulation we generated had at least one \emph{longitudinally shellable} rotation.  Even in the most pathological families we generated, we could find a longitudinally shellable rotation for most triangulations by trying at most four random directions, and a sinkable rotation by trying at most three.\footnote{Our implementation uses the SciPy Python library~\cite{vgohr-sfasc-20}.  In a few extreme cases, SciPy's linear-system solver encountered numerical precision issues, so not all rotations could be classified as sinkable or unsinkable.}

We systematically evaluated two families of random “ugly” triangulations.  The first was generated by computing the convex hull of $100$ random points on the unit sphere, and then attempting to perform a large number of edge flips, each replacing one edge separating a pair of facets whose union is convex.  Specifically, for $10000$ iterations, we chose a random edge and performed a flip if the union of its incident faces is convex; then for an additional $10000$ iterations, we chose a random edge and performed a flip if the union of its incident faces is convex \emph{and} the new edge is longer than the old edge.  We generated $2231$ triangulations in this family and tested $5000$ random directions for each triangulation.  On average, $66.0\%$ of these random directions were shellable and $99.4\%$ were sinkable; in the \emph{worst} triangulation in this family, shown in Figure \ref{F:bad-graph-1}, $17.8\%$ of directions were shellable and $32.9\%$ were sinkable.  

\begin{figure}[htb]
\centering\footnotesize\sffamily
\begin{tabular}{c@{\qquad}c}
	\includegraphics[width=0.4\linewidth]{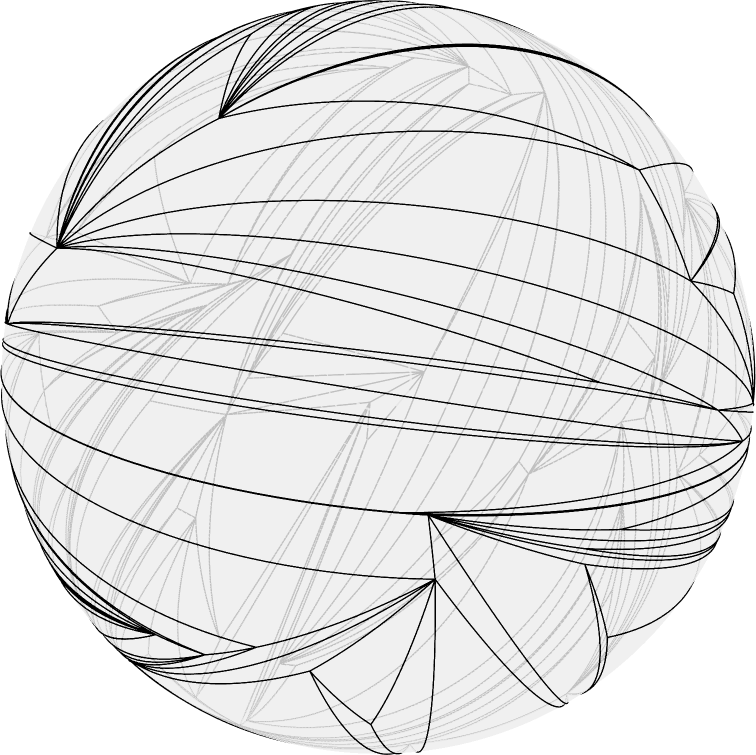} &
	\includegraphics[width=0.4\linewidth]{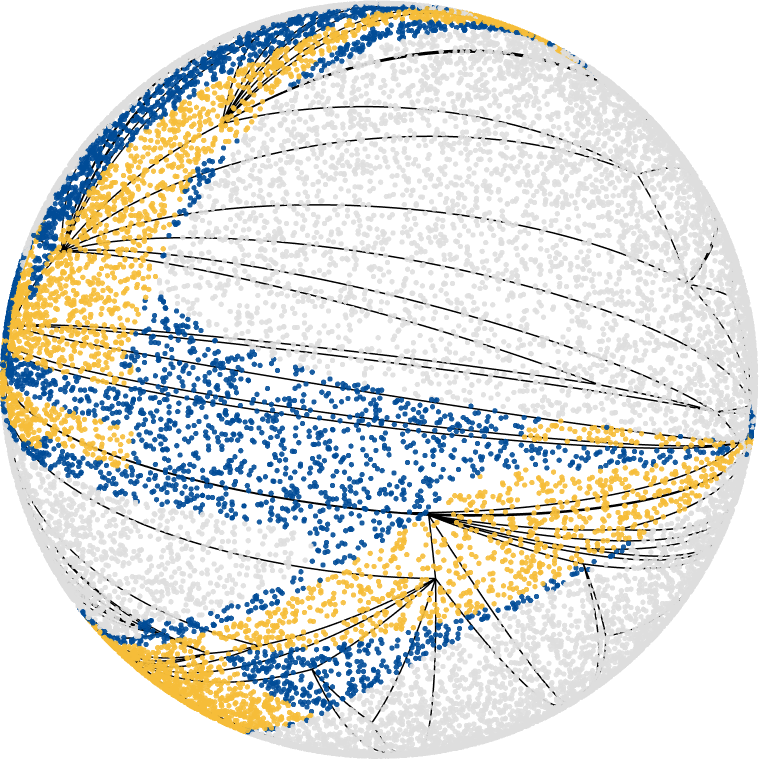} \\
	(a) & (b)
\end{tabular}
\caption{(a)~Our worst triangulation generated from convex hull of random points by randomly flipping edges.  (b)~Classifying random directions as shellable (yellow), not shellable but sinkable (blue), and not sinkable (light gray).}
\label{F:bad-graph-1}
\end{figure}

Our second family was generated by starting with a regular tetrahedron and then adding shortest paths between $100$ pairs of nearly antipodal points as edges.  In each iteration, we chose two random unit vectors $p$ and $r$, defined $q = -p + 0.01\,r$, inserted $p$ and $q$ as new vertices, removed all edges crossing the shortest path from $p$ to $q$, inserted the shortest path $pq$ as an edge, and finally triangulated the spherical polygons on either side of $pq$.  We generated $1346$ triangulations in this family and again tested $5000$ random directions for each triangulation.  On average, $32.6\%$ of these random directions were shellable and $48.3\%$ were sinkable; in the \emph{worst} example in this family, shown in Figure \ref{F:bad-graph-2}, $1.2\%$ of directions were shellable and $1.4\%$ were sinkable.

\begin{figure}[htb]
\centering\footnotesize\sffamily
\begin{tabular}{c@{\qquad}c}
	\includegraphics[width=0.4\linewidth]{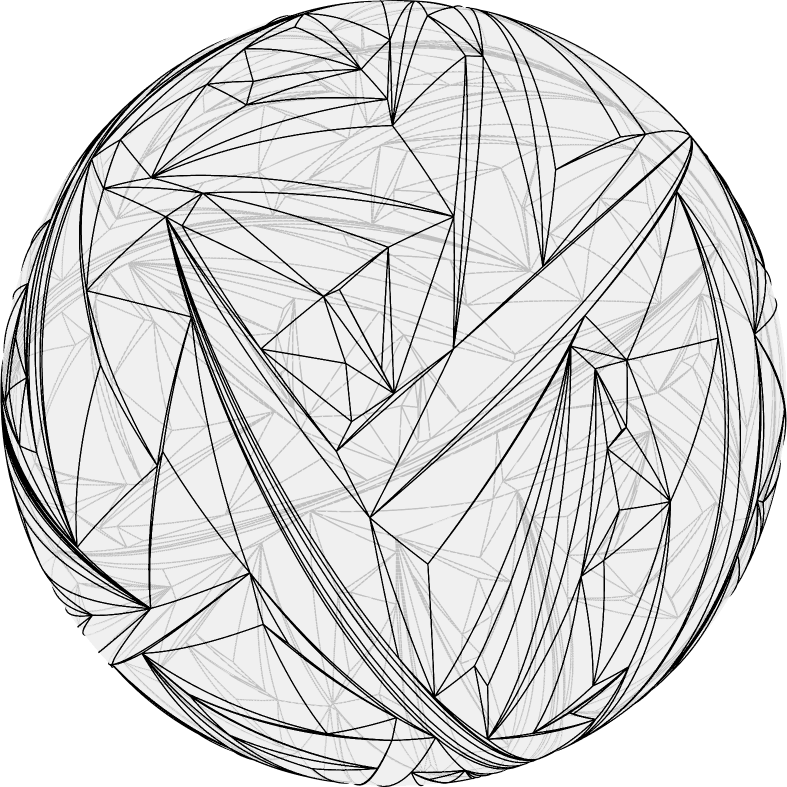} &
	\includegraphics[width=0.4\linewidth]{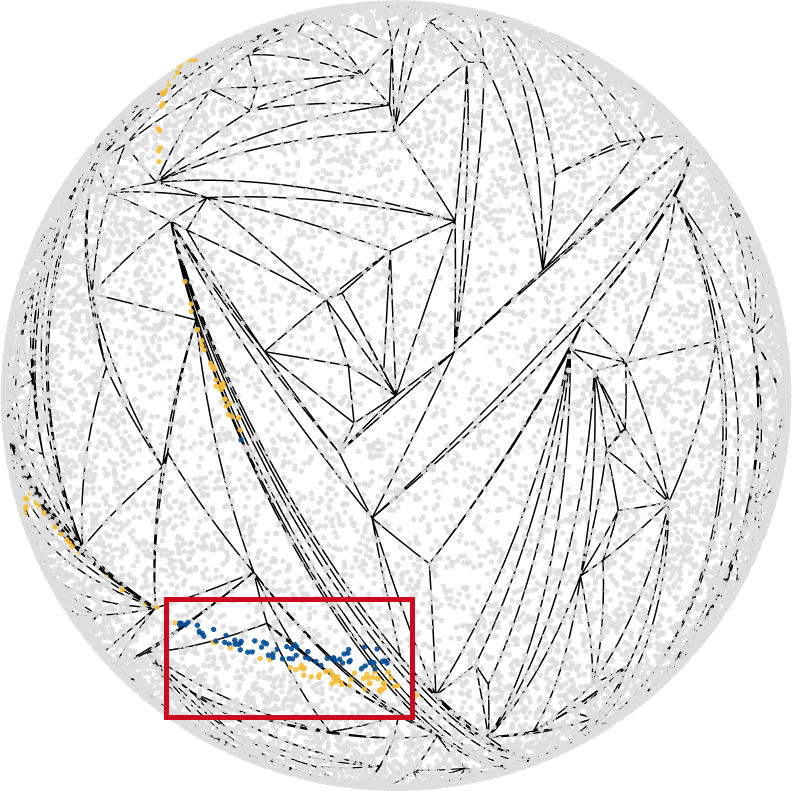} \\
	(a) & (b)
\end{tabular}
\caption{(a)~Our worst triangulation generated by forcing nearly antipodal edges.  (b)~Classifying random directions as shellable (yellow), not shellable but sinkable (blue), and not sinkable (light gray).  Most visible shellable and sinkable directions are inside the red box.}
\label{F:bad-graph-2}
\end{figure}

% ----------------------------------------------------------------------
\section{Conjectures and Open Problems}
\label{S:outro}

Our results suggest several directions for further research.  The most obvious open problem is to prove or disprove Conjecture~\ref{C:sinkrot}: Every shortest-path triangulation of the sphere has a sinkable rotation.  More concretely, is there a polynomial-time algorithm that either finds a sinkable rotation of a given spherical triangulation or correctly reports that no such rotation exists?

A fast algorithm to find sinkable rotations would imply an efficient spherical morphing algorithm, but the resulting morphs have one undesirable feature: Even when we are morphing between full triangulations, most of the intermediate triangulations are not full.  In contrast, Cairns's edge-collapsing argument \cite{c-idgc2-44} actually yields morphs in which every intermediate triangulation is full.  We conjecture that techniques used to efficiently construct piecewise-linear planar morphs, which are based on Cairns's strategy, can be generalized to the sphere to yield \emph{full} morphs.  The most significant obstacle appears to be morphing a shortest-path triangulation with one missing edge, so that the unique quadrilateral face becomes convex.  This task is complicated by the fact that the initial quadrilateral face could have more than one reflex vertex.

Even much simpler questions about sinkability remain open.  Experiments with thousands of random triangulations, similar to those reported in Section \ref{S:experiment}, are consistent with the following conjectures:

\begin{conjecture}
\label{C:acute}
If every edge of a shortest-path triangulation $\Map$ of the sphere has length at most $\pi/2$, then $\Map$ (and therefore every rotation of $\Map$) is sinkable.
\end{conjecture}

\begin{conjecture}
\label{C:belt}
For any shortest-path triangulation~$\Map$, if there is a great circle that does not cross any edge of~$\Map$, then $\Map$ is sinkable.
\end{conjecture}

Either Conjecture \ref{C:acute} or Conjecture \ref{C:belt} would imply the following conjecture about the symmetry of sinkability, which is also consistent with our experimental observations.  A shortest-path triangulation is \EMPH{floatable} if it can be longitudinally morphed into the \emph{northern} hemisphere.

\begin{conjecture}
\label{C:symmetric}
A shortest-path triangulation is sinkable if and only if is it floatable.
\end{conjecture}

The corresponding claim about longitudinal shellability follows directly from Lemma~\ref{L:shellable}; for any unit vector $p$, the graph $\Dn\Map(p)$ is the reversal of $\Dn\Map(-p)$, and the reversal of a directed acyclic graph is a directed acyclic graph.

Two somewhat more technical conjectures concern the linear system \eqref{Eq:Eq} introduced in Theorem \ref{Th:sinksys}.  The only triangulations we have found where this linear system is singular lie on the boundary between sinkable and non-sinkable triangulations.  For example, the linear system associated with the critical Schönhardt triangulation $\bar{S}_{\pi/6}$ is singular, but for all sufficiently small $\e>0$, the linear systems for the sinkable triangulation $\bar{S}_{\pi/6-\e}$ and the unsinkable triangulation~$\bar{S}_{\pi/6+\e}$ are both non-singular.

\begin{conjecture}
If the linear system \eqref{Eq:Eq} defined by a shortest-path triangulation $\Map$ is singular, then $\Map$ is not sinkable. 
\end{conjecture}

\begin{conjecture}
Linear system \eqref{Eq:Eq} is \emph{generically} non-singular; either an arbitrarily small random rotation of~$\Map$ or an arbitrarily small random perturbation of its $x$- and $y$-coordinates yields a nonsingular system \eqref{Eq:Eq} with probability $1$.
\end{conjecture}

Ho \cite{h-chps2-73,h-chps1-73} proved that the space of all embeddings of a maximal planar graph, with a fixed outer triangular face, is topologically trivial.  Ho's result was generalized to planar triangulations with convex outer faces by Bloch, Connelly, and Henderson \cite{bch-sslhc-84} and Bloch \cite{b-scsle-85}; more recent proofs have been given by Cerf \cite{c-abcht-19} and Luo~\cite{l-sgts-22}.  Connelly, Henderson, Ho, and Starbird~\cite{chhs-prlhe-83} conjectured that every isotopy class of geodesic triangulations on \emph{any} surface $S$ with constant curvature is homotopy-equivalent to the group $\emph{Isom}_0(S)$ of isometries of~$S$ that are homotopic to the identity.  Their conjecture was recently proved both for the flat torus \cite{lwz-dsgtf-21,el-ptmme-23} and for arbitrary negative-curvature surfaces \cite{lwz-dsgtg-23}; all of these proofs rely on nontrivial extensions of Tutte and Floater's planar barycentric embedding theorem~\cite{c-crgtd-91,ggt-domam-06,hs-seshm-15,lwz-dsgtg-23}.  The only closed surface for which this conjecture remains open is the sphere!  Cairns actually announced a proof of this conjecture for the sphere in 1941 \cite{c-svgcs-41} but later retracted it \cite{c-idgc2-44,c-dprc-44}; Awartani and Henderson also posed this conjecture as a motivation for their work \cite{ah-sgts-87}.

Finally, essentially nothing is known about morphing the more general class of \emph{long geodesic embeddings}, whose edges are great-circular arcs that may be longer than a semicircle.  Even representing such embeddings is nontrivial; in particular, vertex coordinates and a rotation system do not necessarily specify a unique long geodesic embedding.  Using ad-hoc arguments, we have been able to prove that any two long geodesic embeddings of $K_4$ are connected by a continuous family of such embeddings.  Do such morphs exist for all planar graphs?  Can they be constructed efficiently?

\paragraph*{Acknowledgements.}  The authors thank Eli Kujawa for early helpful discussions, the anonymous reviewers of the SOCG version of this paper for valuable comments and suggestions, and Danny Halperin for sharing a copy of his unpublished manuscript \cite{h-esptu-08}.

% ======================================================================
%\newpage
\bibliographystyle{jeffe}
\bibliography{topology,compgeom,optimization}

\end{document}